\documentclass[11pt]{article}

\usepackage{amsmath}
\usepackage{amsfonts}
\usepackage{amssymb}
\usepackage{graphicx}
\usepackage{subfigure}
\usepackage{rotating}
\usepackage{color}
\usepackage{amsthm}
\usepackage{authblk}
\usepackage{natbib}
\usepackage{threeparttable}
\usepackage{pdflscape}
\usepackage{multirow}
\usepackage{accents}
\usepackage{longtable}
\usepackage{setspace}

\allowdisplaybreaks[3]

\newtheorem{thm}{Theorem}[section]

\newtheorem{lem}[thm]{Lemma}
\newtheorem{proposition}[thm]{Proposition}

\def\max{\mathop{\rm max}}

\begin{document}
\setlength{\textheight}{575pt}
\setlength{\baselineskip}{23pt}

\title{Integrative Analysis of Prognosis Data on Multiple Cancer Subtypes using Penalization}
\author[1]{Jin Liu}
\affil[1]{School of Public Health, Yale University}
\author[2]{Jian Huang}
\affil[2]{Department of Statistics $\&$ Actuarial Science, Department of Biostatistics, University of Iowa}
\author[1]{Yawei Zhang}
\author[3]{Qing Lan}
\affil[3]{Division of Cancer Epidemiology and Genetics, National Cancer Institute, NIH}
\author[3]{Nathaniel Rothman}
\author[1]{Tongzhang Zheng}
\author[1]{Shuangge Ma}

\maketitle
\def\ep{\varepsilon}
\def\ba{{\boldsymbol a}}
\def\bA{{\boldsymbol A}}
\def\cA{{\cal A}}
\def\hA{\widehat{A}}
\def\tba{\widetilde{\boldsymbol a}}
\def\bb{{\boldsymbol b}}
\def\hb{\hat{b}}
\def\hbb{\hat{\boldsymbol b}}
\def\tbb{\tilde{\boldsymbol b}}
\def\barb{\bar{b}}
\def\bbb{\overline{\bb}}
\def\Bbar{\overline{B}}
\def\bB{{\boldsymbol B}}
\def\cB{{\cal B}}
\def\hcB{{\widehat{\cal B}}}
\def\bC{{\boldsymbol C}}
\def\bD{{\boldsymbol D}}
\def\bd{{\boldsymbol d}}
\def\be{{\boldsymbol e}}
\def\cE{{\cal E}}
\def\rE{{\mathrm E}}
\def\bff{{\boldsymbol f}}
\def\hf{\widehat{f}}
\def\cF{{\cal F}}
\def\bg{{\boldsymbol g}}
\def\bG{{\boldsymbol G}}
\def\cG{{\cal G}}
\def\bh{{\boldsymbol h}}
\def\cH{{\cal H}}
\def\bI{{\boldsymbol I}}
\def\bell{{\boldsymbol \ell}}
\def\tM{\widetilde{M}}
\def\bO{{\boldsymbol O}}
\def\cO{{\cal O}}
\def\bp{{\boldsymbol p}}
\def\rP{{\mathrm P}}
\def\bP{{\boldsymbol P}}
\def\bbP{{\mathbb P}}
\def\tP{\widetilde{P}}
\def\bQ{{\boldsymbol Q}}
\def\bfr{{\boldsymbol r}}
\def\tr{\tilde{r}}
\def\hr{\hat{r}}
\def\hbr{\hat{\boldsymbol r}}
\def\tbr{\tilde{\boldsymbol r}}
\def\bs{{\boldsymbol s}}
\def\hs{\widehat{ s}}
\def\cS{{\cal S}}
\def\bt{{\boldsymbol t}}
\def\bT{{\boldsymbol T}}
\def\bu{{\boldsymbol u}}
\def\hbu{\widehat{\boldsymbol u}}
\def\bU{{\boldsymbol U}}
\def\tu{\tilde{u}}
\def\bv{{\boldsymbol v}}
\def\bV{{\boldsymbol V}}
\def\bw{{\boldsymbol w}}
\def\bW{{\boldsymbol W}}
\def\tw{\tilde{w}}
\def\bx{{\boldsymbol{x}}}
\def\tx{\widetilde{x}}
\def\tbx{\widetilde{\boldsymbol{x}}}
\def\bX{{\boldsymbol X}}
\def\cX{{\cal X}}
\def\tX{\widetilde{X}}
\def\ty{\tilde{y}}
\def\by{{\boldsymbol y}}
\def\bY{{\boldsymbol Y}}
\def\tby{\tilde{\boldsymbol y}}
\def\tY{\widetilde{Y}}
\def\hy{\hat{y}}
\def\byhat{{\hat {\boldsymbol y}}}
\def\tY{\widetilde{Y}}
\def\cY{{\cal Y}}
\def\Ybar{\overline{Y}}
\def\bz{{\boldsymbol z}}
\def\bZ{{\boldsymbol Z}}
\def\cZ{{\cal Z}}
\def\tz{\tilde{z}}
\def\tZ{\widetilde{Z}}
\def\brho{{\boldsymbol{\rho}}}
\def\bzero{{\boldsymbol 0}}
\def\eps{\epsilon}
\def\veps{\varepsilon}
\def\bveps{\boldsymbol{varepsilon}}
\def\tveps{\widetilde{\varepsilon}}
\def\tbveps{\widetilde{\boldsymbol{\varepsilon}}}
\def\Ghat{\widehat{G}}
\def\argmax{\mathop{\rm argmax}}
\def\argmin{\mathop{\rm argmin}}
\def\real{\mathop{{\rm I}\kern-.2em\hbox{\rm R}}\nolimits}
\def\diag{\mbox{diag}}

\def\sgn{\hbox{sgn}}
\def\Var{\hbox{Var}}
\def\Cov{\hbox{Cov}}
\def\Rem{\hbox{Rem}}

\def\whbeta{\widehat{\beta}}
\def\hbeta{\hat{\beta}}
\def\whtheta{\widehat{\theta}}
\def\htheta{\hat{\theta}}
\def\whF{\widehat{F}}
\def\hF{\hat{F}}
\def\Dnu{\Delta_{\nu}}
\def\median{\mbox{median}}
\def\sign{\mbox{sign}}
\def\trace{\mbox{trace}}

\def\bone{{\boldsymbol 1}}
\def\bzero{{\boldsymbol 0}}
\def\balpha{\boldsymbol \alpha}
\def\btheta{\boldsymbol \theta}
\def\bbeta{\boldsymbol \beta}
\def\bgamma{\boldsymbol \gamma}
\def\hbbeta{\hat{\boldsymbol \beta}}
\def\hsbbetan{\hat{\boldsymbol \beta_n^*}}
\def\tbeta{\tilde{\beta}}
\def\tbbeta{\tilde{\boldsymbol \beta}}
\def\bdelta{\boldsymbol \delta}
\def\bata{\boldsymbol \eta}
\def\gam{\gamma}
\def\lam{\lambda}
\def\blam{\boldsymbol \lambda}
\def\hmu{\widehat{\mu}}
\def\bmu{\boldsymbol \mu}
\def\bnu{\boldsymbol \nu}
\def\hphi{\widehat{\phi}}
\def\drho{\dot{\rho}}
\def\hsigma{\widehat{\sigma}}
\def\ttheta{\widetilde{\theta}}
\def\hbtheta{\widehat{\boldsymbol \theta}}
\def\bveps{\boldsymbol \varepsilon}
\def\tbveps{{\tilde\bveps}}
\def\bxi{\boldsymbol \xi}
\def\txi{\tilde{\xi}}
\def\tzeta{\tilde{\zeta}}

\begin{abstract}
In cancer research, profiling studies have been extensively conducted, searching for genes/SNPs associated with prognosis. Cancer is a heterogeneous disease. Examining similarity and difference in the genetic basis of multiple subtypes of the same cancer can lead to better understanding of their connections and distinctions. Classic meta-analysis approaches analyze each subtype separately and then compare analysis results across subtypes. Integrative analysis approaches, in contrast, analyze the raw data on multiple subtypes simultaneously and can outperform meta-analysis. In this study, prognosis data on multiple subtypes of the same cancer are analyzed.  An AFT (accelerated failure time) model is adopted to describe survival. The genetic basis of multiple subtypes is described using the heterogeneity model, which allows a gene/SNP to be associated with the prognosis of some subtypes but not the others. A compound penalization approach is developed to conduct gene-level analysis and identify genes that contain important SNPs associated with prognosis. The proposed approach has an intuitive formulation and can be realized using an iterative algorithm. Asymptotic properties are rigorously established. Simulation shows that the proposed approach has satisfactory performance and outperforms meta-analysis using penalization. An NHL (non-Hodgkin lymphoma) prognosis study with SNP measurements is analyzed. Genes associated with the three major subtypes, namely DLBCL, FL, and CLL/SLL, are identified. The proposed approach identifies genes different from alternative analysis and has reasonable prediction performance.
\end{abstract}

\section{Introduction}

Profiling studies have been extensively conducted in cancer research, searching for SNPs (single nucleotide polymorphisms) and genes that are associated with prognosis. Cancer is a heterogeneous disease. Different subtypes of the same cancer usually have different prognosis patterns and different associated genes/SNPs. Consider NHL (non-Hodgkin lymphoma), which is a heterogeneous group of malignancies ranging from very indolent forms to aggressive ones. As discussed in Zhang et al. (2011), different subtypes of NHL are largely different. For example, DLBCL, the largest subtype, is aggressive, whereas FL, the second largest subtype, is indolent. Chromosomal translocations such as t(3, 22) are specific to DLBCL, whereas others such as t(14, 18) are specific to FL. On the other hand, different subtypes may also share common susceptibility genes/SNPs. Genes in the cell cycle, multiple signaling, RAS, and DNA repair pathways are involved in the development and progression of multiple cancers including NHL. For NHL, Han X. et al. (2010) and Ma et al. (2010) have found that SNPs in multiple genes, such as BRCA2, CASP3, IRF1, BCL2, NAT2, ALXO12B, are associated with both DLBCL and FL. Investigating the similarity and difference in the genetic basis of multiple subtypes of the same cancer can lead to better understanding of the connections and distinctions among subtypes (Rhodes et al. 2004; Goh and Choi 2012).

A comparable setting where the discussion and proposed method are relevant is the analysis of prognosis data on multiple types of cancers. As discussed in Rhodes et al. (2004) and followup studies, susceptibility genes shared by multiple types of cancers are more likely to represent the more essential features of cancer, whereas cancer type-specific genes determine the distinctions among different cancers.

When multiple subtypes of the same cancer are of interest, as discussed in Zhang et al. (2011), most studies analyze each subtype separately and then compare results across subtypes. For NHL, Han et al. (2010) and Ma et al. (2010) take this approach. Such a strategy fits in the classic meta-analysis framework. With high-dimensional measurements such as SNPs, data on individual subtypes have the ``large $d$, small $n$" characteristic, with the sample size $n$ much smaller than the number of SNPs $d$. Because of the low sample size, susceptibility genes/SNPs identified from the analysis of each subtype may have unsatisfactory properties. Recent studies have shown that, when multiple datasets (multiple subtypes in this study) have overlapping susceptibility SNPs/genes, integrative analysis can analyze raw data of multiple datasets simultaneously and generate improved analysis results over the analysis of individual datasets and meta-analysis (Liu et al. 2012; Ma et al. 2009; Ma et al. 2012).

With data on multiple subtypes, the goal is to identify genes associated with prognosis. For marker identification, we adopt penalization, which has been extensively applied to the analysis of cancer prognosis data with high-dimensional genetic measurements. Single-dataset penalization methods, such as Lasso, SCAD, bridge, MCP and their group counterparts, cannot be directly applied to the analysis of multiple datasets. With multiple datasets, the homogeneity model assumes that,
if a functional unit (gene or SNP) is identified, it is concluded as associated with prognosis in all datasets (Liu et al. 2012). An alternative to the homogeneity model is the heterogeneity model, under which a gene or SNP can be associated with prognosis in some datasets but not the others. Under the heterogeneity model, research on penalization methods has been limited (Liu et al. 2012). Compared with the existing studies which analyze gene expression data, the present one has additional complexity. One gene may consist of multiple SNPs, and it is important to allow different effects for SNPs within the same gene. In addition, theoretical properties of the methods developed in Liu et al. (2012) and others have not been established. To the best of our knowledge, the only available method that is tailored to the type of data analyzed in this study is Ma et al. (2012), which adopts thresholding for marker selection. The thresholding method does not have a well-defined objective function. Thus its properties can be very difficult to establish. In addition, it may have more tuning parameters than the penalization method. Compared with the existing studies, another advancement of this study is the analysis of a prognosis data on NHL, which may provide insights into the genetic basis of this deadly disease.

The integrative analysis of data on multiple subtypes of cancer can be challenging. With some cancers, the subtype information may be only partial or even wrong. In addition, the definitions of subtypes are still evolving. For NHL subtypes, we refer to Zhang et al. (2011) and references therein for relevant discussions. When there are a large number of subtypes, the set of subtypes chosen for analysis needs to be jointly determined by the scientific question of interest, quality of data, sample size, evidence from epidemiologic studies and other factors. We acknowledge the importance and difficulty of these issues. In this study, we focus on the development of a new analysis approach and refer to other publications for relevant discussions.

\section{Integrative Analysis under the Heterogeneity Model}

\subsection{Data and model settings}

Assume that there are data on $M$ subtypes of the same cancer, and there are $n^m$ iid observations for subtype $m(=1,\ldots, M)$. The total sample size is $n=\sum_m n^m$. For subtype $m$, denote $T^m$ as the logarithm of failure time. Denote $X_o^m$ as the length-$d$ covariate vector (SNPs in this study). The subscript ``o" is used to discriminate the original (versus weighted) covariates.
For simplicity of notation, assume that different subtypes measure the same set of covariates. In practice, if a covariate is not measured for a specific subtype, its corresponding regression coefficient will be set as zero. In penalization, rescaling is used to accommodate partially matched covariate sets.

For subject $i$ of subtype $m$, the AFT (accelerated failure time) model assumes that
\begin{equation}
T_i^m=\beta_0^m + X_{oi}^{m\prime}\beta^m+\epsilon_i^m, ~i = 1,\dots,n^m.
\end{equation}
where $\beta_0^m$ is the intercept, $\beta^m\subseteq\mathbb{R}^d$ is the regression coefficient, and $\epsilon_i^m$ is the error term. As $T_i^m$ is subject to right censoring, we observe $(Y_{oi}^m,\delta_i^m,X_{oi}^m)$, where $Y_{oi}^m=\mathrm{min}\{T_i^m,C_i^m\}$, $C_i^m$ is the logarithm of censoring time, and $\delta_i^m=I\{T_i^m\le C_i^m\}$ is the event indicator.

Compared with alternatives such as the Cox model, the AFT model has a much simpler objective function, as demonstrated in the next section, and hence significantly lower computational cost. Such a property is particularly desirable for high-dimensional data. In addition, it directly describes event times, and its regression coefficients may have more lucid interpretations than those in alternative models. As there is a lack of model diagnostics tools for high-dimensional data, alternative models will not be discussed.

\subsection{Weighted least squares estimation}

Let $\hat{F}^m$ be the Kaplan-Meier estimator of the distribution function $F^m$ of $T^m$. Following Stute (1996), $\hat{F}^m$ can be written as $\hat{F}^m (y)=\sum_{i=1}^{n^m}\omega_i^m 1\{Y_{o(i)}^m\le y\}$, where the $\omega_i^m$'s are the jumps in the Kaplan-Meier estimator and can be expressed as
\begin{equation}
\omega_1^m=\frac{\delta_{(1)}^m}{n^m},~
\omega_i^m=\frac{\delta_{(i)}^m}{n^m-i+1}\prod_{j=1}^{i-1}\left( \frac{n^m-j}{n^m-j+1}\right)^{\delta_{(j)}^m},i=2,\dots,n^m. \notag
\end{equation}
$Y_{o(1)}^m\le\cdots\le Y_{o(n^m)}^m$ are the order statistics of $Y_{oi}^m$'s, and $\delta_{(1)}^m,\dots,\delta_{(n^m)}^m$ are the associated event indicators. Similarly, let $X_{o(1)}^m,\dots,X_{o(n^m)}^m$ be the associated covariate vectors of the ordered $Y_i^m$'s. Consider the weighted least squares objective function
\begin{equation}
L^m(\beta_0^m, \beta^m)=\frac{1}{2}\sum_{i=1}^{n^m} \omega_i^m \left(Y_{o(i)}^m-\beta_0^m-{X_{o(i)}^m}^{\prime}\beta^m\right)^2.
\end{equation}

Let $\bar{X}_{\omega}^m=\sum_{i=1}^{n^m} \omega_i^m X_{o(i)}^m/\sum_{i=1}^{n^m} \omega_i^m$, $\bar{Y}_{\omega}^m=\sum_{i=1}^{n^m} \omega_i^m Y_{o(i)}^m/\sum_{i=1}^{n^m} \omega_i^m$, $X_{\omega(i)}^m=(\omega_i^m)^{1/2}(X_{o(i)}^m-\bar{X}_{\omega}^m)$, and $Y_{\omega(i)}^m=(\omega_i^m)^{1/2}(Y_{o(i)}^m-\bar{Y}_{\omega}^m)$. Using the weighted centered values, the intercept is zero. The weighted least squares objective function can be written as
\begin{equation}
L^m(\beta^m)=\frac{1}{2}\sum_{i=1}^{n^m} \left(Y_{\omega(i)}^m-{X_{\omega(i)}^m}^{\prime}\beta^m\right)^2.
\end{equation}
This simple form makes computation affordable even with high-dimensional data.

Assume independence between data for the $M$ subtypes. Consider the overall objective function
$$L(\beta)=\sum_{m=1}^M \frac{1}{n^m} L^m(\beta^m).$$
Here we normalize $L^m$ by $n^m$ so that the analysis is not dominated by large subtypes. When larger subtypes are of more interest, the unnormalized objective function may be considered.

\subsection{Heterogeneity model}
As formulated in Liu et al. (2012), two different models, namely the homogeneity model and heterogeneity model, can be applied to describe the genetic basis of multiple subtypes. Under the homogeneity model, it is postulated that multiple subtypes share the same set of susceptibility SNPs/genes. Considering the significantly different prognosis patterns of different subtypes, this model may be too restricted. Under the heterogeneity model, the sets of susceptibility SNPs/genes may differ across subtypes. The heterogeneity model includes the homogeneity model as a special case and can be more flexible.

To more explicitly describe the data and model settings, heterogeneity model, and our analysis strategy, consider a hypothetical cancer study with three subtypes and eight SNPs representing four genes (Table 1). Gene 1 is associated with the prognosis of all three subtypes; Gene 2 is associated with the first two subtypes but not the third one; Gene 3 is associated with the third subtype only; And gene 4 is not associated with any subtype. In Table 1, we show the regression coefficient matrix whose main characteristics reflect the essence of integrative analysis under the heterogeneity model. Unimportant genes/SNPs not associated with prognosis have no effects and so zero regression coefficients. With penalization approaches including the proposed one, marker identification amounts to identifying the sparsity structure of models. For an important gene/SNP (for example SNP 1\_1), its strengths of association with multiple subtypes, which are measured with regression coefficients, may be different for different subtypes. With SNP data, analysis can be conducted at multiple levels, particularly including SNP-level and gene-level. In this study, we conduct gene-level analysis, which complements SNP-level analysis. As the goal is to identify important genes that contain prognosis-associated SNPs, within an important gene, no further selection is conducted. Thus, SNPs within the same gene have the ``all in or all out" property. Such a strategy has been adopted in Ma et al. (2012) and is different from that for SNP-level analysis approaches.

\section{Penalized Marker Selection}
\label{penalty_sec}

Penalization is adopted for marker selection. Under the data and model settings described in the previous section, existing penalization approaches are not directly applicable. A new penalization approach is described in Section \ref{pls}. An effective computational algorithm is proposed in Section ~\ref{comp}. Tuning parameter selection is discussed in Section \ref{tuning}. Some practical concerns are discussed in Section \ref{practical}. Asymptotic properties and proofs are established in Appendix.

\subsection{Penalty function}
\label{pls}

Assume that the $d$ SNPs belong to $J$ genes. To accommodate the scenario with partially matched gene sets, suppose that $M_j$ subtypes (studies) measure gene $j$. Without loss of generality, assume that gene $j$ is measured in the first $M_j$ subtypes. Denote $d_{jk}$ as the number of SNPs in the $j$th gene and $k$th subtype  with coefficient vector $\beta_{jk}=(\beta_{jk}^{1},\dots,\beta_{jk}^{d_{jk}})^{\prime}$, $j=1,\dots,J$, $k=1,\dots,M_j$. The subscript $k$ is kept to accommodate partially matched SNP sets for the same genes. Then $\beta_j=(\beta_{j1}^{\prime},\dots,\beta_{jM_j}^{\prime})^{\prime}$ is the regression coefficient for all SNPs in the $j$th gene across all subtypes.  Here the notations are slightly more complicated than those in Section 2 to accommodate the ``SNP-within-gene" hierarchical structure and partially matched SNP/gene sets. $\beta=(\beta_1^{\prime}, \ldots, \beta_J^{\prime})^{\prime}$.

Consider the penalized estimate
$$\hat\beta=argmin\left\{ L(\beta)+P_{\lambda_n,\gamma}(\beta)\right\}.$$
A nonzero component of $\hat\beta$ indicates an association between the corresponding gene (SNP) and subtype's prognosis. Consider the penalty function
\begin{equation}
P_{\lambda_n,\gamma}(\beta)=\lambda_n\sum_{j=1}^J c_j \left( \sum_{k=1}^{M_j} \sqrt{d_{jk}} \|\beta_{jk}\| \right)^{\gamma}, \label{penalty}
\end{equation}
where $\lambda_n> 0$ is a data-dependent tuning parameter, $c_j\propto M_j^{1-\gamma}$ is a constant accommodating partially matched gene sets, $||\cdot||$ is the $L_2$ norm, and $0<\gamma<1$ is the fixed bridge parameter.

The above penalty has been designed to tailor the special model characteristics as described in Table 1. In our analysis, genes are the basic functional units. The penalty is the sum of $J$ individual terms, with one for each gene. For a specific gene, two levels of selection need to be conducted. The first is to determine whether it is associated with any subtype at all. This step of selection is achieved using a bridge penalty. For a gene associated with at least one subtype, the second level of selection is to determine which subtype(s) it is associated with. This step of selection is achieved using a Lasso-type penalty. The composition of the bridge-type penalty and the Lasso-type penalty can achieve the desired two-level selection. One gene may contain multiple SNPs. The effect of gene $j$ for subtype $k$ is represented by the vector $\beta_{jk}$. Here the penalty is imposed on the $L_2$ norm of $\beta_{jk}$, which can be viewed as the square root of a ridge penalty. Thus, within a selected gene, no further SNP-level selection is conducted. If a gene is selected, all SNPs within this gene are selected.

When $d_{jk}\equiv 1$ (one SNP per gene), penalty (\ref{penalty}) becomes the group bridge, which has been developed for the analysis of a single dataset in Huang et al. (2009). This study is among the first to apply group bridge type penalization in integrative analysis. In addition, the proposed penalized estimation can be more complicated than that in Huang et al. (2009) by accommodating the ``SNP-within-gene" hierarchical structure. Using composite penalization for marker selection under the heterogeneity model has been proposed in Liu et al. (2012) for diagnosis studies with binary responses. The penalty in Liu et al. (2012) is built on the composition of MCP and Lasso, which is computationally more expensive, and cannot accommodate the ``SNP-within-gene" structure.

\subsection{Computational algorithm}
\label{comp}

For subtype $m$, denote $Y^m$ as the vector composed of $\frac{\sqrt{n}}{\sqrt{n^m}}Y^m_{\omega}$s, and $X^m$ as the matrix composed of $\frac{\sqrt{n}}{\sqrt{n^m}} X^m_{\omega}$s. Let $Y=(Y^{1\prime},\dots,Y^{M\prime})^{\prime}$ and $X=\mathrm{diag}(X^1,\dots,X^M)$.
Denote $X_j$ as the submatrix of $X$ corresponding to $\beta_j$, and its dimension is $n\times \sum_{k=1}^{M_j} d_{jk}$. Then
\begin{equation}
 L(\beta)=\sum_{m=1}^M \frac{1}{2n^m}\sum_{i=1}^{n^m} (Y_{\omega}^m-{X_{\omega}^m}^{\prime}\beta^m)^2=\frac{1}{2n}\lVert Y-\sum_{j=1}^J X_j\beta_j\rVert^2.
\end{equation}

The overall objective function is
\begin{equation}
\frac{1}{2n}\lVert Y-\sum_{j=1}^J X_j\beta_j\rVert^2+\lambda_n\sum_{j=1}^J c_j \left( \sum_{k=1}^{M_j} \sqrt{d_{jk}} \|\beta_{jk}\| \right)^{\gamma}.  \label{obj}
\end{equation}
Define
\begin{equation}
\label{pls_s}
S (\beta,\theta)=\frac{1}{2n}\lVert Y-\sum_{j=1}^J X_j\beta_j\rVert^2+\sum_{j=1}^J \theta_j^{1-1/\gamma}c_j^{1/\gamma}\sum_{k=1}^{M_j}
\sqrt{d_{jk}}\|\beta_{jk}\|+\tau_n\sum_{j=1}^J\theta_j,
\end{equation}
where $\theta=(\theta_1, \ldots, \theta_J)^{\prime}$, and $\tau_n$ is a penalty parameter.

\begin{proposition}
\label{prop1}
If $\lambda_n = \tau_n^{1-\gamma}\gamma^{-\gamma}(1-\gamma)^{\gamma-1}$,
then $\hat\beta$ minimizes the objective function in (\ref{obj}) if and only if $(\hat\beta,\hat\theta)$ minimizes $S(\beta,\theta)$ subject to $\theta_j\ge 0$ for all $j$.
\end{proposition}
Proof is provided in Appendix. Examining $S(\beta, \theta)$ suggests that optimization with respect to $\beta$ and $\theta$ can be conducted ``separately". Optimization with respect to $\theta$ has a simple analytic solution. Optimization with respect to $\beta$ has a weighted group Lasso type objective function, for which there exist effective algorithms. Motivated by such an observation, consider the following algorithm:
\begin{enumerate}
\item
Denote $\beta^{(0)}$ as the initial estimate. In our limited numerical study, the choice of initial estimate does not have much impact on the final estimate. For simplicity, all components of $\beta^{(0)}$ are set to be 1. Set $s=0$;
\item
$s=s+1$. Compute
\begin{equation}
 \theta_j^{(s)}=c_j\left( \frac{1-\gamma}{\gamma\tau_n} \right)^\gamma \left( \sum_{k=1}^{M_j} \sqrt{d_{jk}}\left\Vert\beta_{jk}^{(s-1)}\right\Vert \right)^{\gamma},
\end{equation}
\begin{equation}
\label{update2}
\beta^{(s)}=argmin_{\beta}\left( \frac{1}{2n}\lVert Y-\sum_{j=1}^J X_j\beta_j\rVert^2+\sum_{j=1}^J (\theta_j^{(s)})^{1-1/\gamma}c_j^{1/\gamma}\sum_{k=1}^{M_j}
\sqrt{d_{jk}}\|\beta_{jk}\| \right),
\end{equation}

\item
Repeat Step 2 until convergence.
\end{enumerate}

In numerical study, we use the $L_2$ norm of the difference between two consecutive estimates less than 0.001 as the convergence criterion. The proposed algorithm always converges, since at each step, the nonnegative objective function decreases. It is noted that as the group bridge type penalty is not convex, the algorithm may converge to a local minimizer depending on the initial value $\beta^{(0)}$. Using the proposed initial value works well in our numerical study. The main computational task is the computation of $\beta^{(s)}$, which is a group Lasso type estimate and can be achieved using a group coordinate descent algorithm (Huang et al. 2012; Liu et al. 2012). Convergence of the group coordinate descent algorithm can be derived following Tseng (2001).

\subsection{Tuning parameter selection}
\label{tuning}

The proposed penalty involves two tuning parameters $\gamma$ and $\lambda_n$. In the study of bridge type penalties (Huang et al. 2009), the value of $\gamma$ is usually fixed. Theoretically speaking, different values of $\gamma$, as long as in the interval (0, 1), lead to similar asymptotic results. In practice as $\gamma\to 1$, the bridge type penalty goes to the Lasso type penalty; On the other hand, as $\gamma \to 0$, it behaves similarly to AIC/BIC type penalties. In our numerical study, we experiment with a few $\gamma$ values, particularly including 0.5, 0.7 and 0.9. The effect of $\lambda_n$ is similar to that with other penalties. As $\lambda_n \to \infty$, fewer genes/SNPs are identified.

As the function $L(\beta)$ has a least squares form, we propose using BIC for tuning parameter selection. Particularly, with a fixed $\gamma$, the optimal $\lambda_n$ minimizes
\begin{equation}
\mathrm{BIC}(\lambda_n) = \log \left\{ ||Y-X\hat{\beta}(\lambda_n)||^2/n\right\}+\log (n)\mathrm{df}(\lambda_n)/n. \notag
\end{equation}
Here we use the notation $\hat{\beta}(\lambda_n)$ to emphasize the dependence of $\hat\beta$ on $\lambda_n$. Motivated by Yuan and Lin (2006), an approximation of the degree of freedom is adopted as
\begin{equation}
\label{df}
\tilde{\mathrm{df}}(\lambda_n)=\sum_{j=1}^J\sum_{k=1}^{M_j} I(||\hat\beta_{jk}||>0)+\sum_{j=1}^J\sum_{k=1}^{M_j} \frac{||\hat\beta_{jk}||}
{||\hat\beta_{jk}^{\mathrm{LS}}||}(d_{jk}-1).
\end{equation}
Here $\hat\beta_{jk}^{\mathrm{LS}}$  is obtained by fitting an AFT model (with least squares estimation) using the $j$th gene and $k$th subtype only.

\subsection{Practical considerations}
\label{practical}
With practical data, minor allele frequencies in some loci can be low. This may cause an instability problem in the Cholesky decomposition when some eigenvalues of the correlation matrices are too small. In the proposed penalized selection, within-gene-SNP level selection is not of interest. To reduce the dimensionality within genes and to tackle the colinearity problem, when there is evidence of a lack of stability, we first conduct principal component analysis (PCA) within genes. Specifically, we choose the number of PCs such that at least 90$\%$ of the total variation is explained. Then the PCs, as opposed to the original SNP measurements, are used in downstream analysis. Our empirical study suggests that this simple step may ensure that the smallest eigenvalues of the covariance matrices are not too small and that the Cholesky decomposition is stable.

\section{Simulation Study}
Three datasets (subtypes) are simulated, each with 100 subjects. For each subject, the genotypes of 200 genes are simulated, each with 5 SNPs. There are thus a total of 1000 SNPs for each subtype. The genotypes are first generated from multivariate normal distributions. Then the value of each SNP is set equal to 0, 1, or 2, depending on whether the continuous (normally distributed) value is $<-c$, $\in[-c, c]$, or $>c$, where $c$ is the 3$rd$ quartile of the standard normal distribution. Thus on average, each SNP has equal allele frequencies.
Genotype $j$ and $k$, if from different genes, have correlation coefficient $0.2^{|j-k|}$. For genotypes from the same genes, consider the following two correlation structures. The first is the auto-regressive correlation, where genotype $j$ and $k$ have correlation coefficient $\rho^{|j-k|}$. $\rho=0.2$, 0.5, and 0.8, corresponding to weak, moderate, and strong correlations, respectively. The second is the banded correlation structure. Here three scenarios are considered. Under the first scenario, genotype $j$ and $k$ have correlation coefficient 0.2 if $|j-k|=1$, 0.1 if $|j-k|=2$, and 0 otherwise; Under the second scenario, genotype $j$ and $k$ have correlation coefficient 0.5 if $|j-k|=1$, 0.25 if $|j-k|=2$, and 0 otherwise. Under the third scenario, genotype $j$ and $k$ have correlation coefficient 0.6 if $|j-k|=1$, 0.33 if $|j-k|=2$, and 0 otherwise.

Consider two cases of the nonzero regression coefficients. In case 1, the nonzero regression coefficients for subtype 1 and 2 are (0.15, 0.15, 0.15, 0.15, 0.15, 0.1, 0.1, 0.1, 0.1, 0.1, 0.15, 0.15, 0.15, 0.15, 0.15, 0.1, 0.1, 0.1, 0.1, 0.1), and the nonzero coefficients for subtype 3 are (0.1, 0.1, 0.1, 0.1, 0.1, 0.15, 0.15, 0.15, 0.15, 0.15, 0.15, 0.15, 0.15, 0.15, 0.15, 0.1, 0.1, 0.1, 0.1, 0.1). In case 2, the nonzero regression coefficients for subtype 1 and 2 are (0.15, 0.15, 0.15, 0.15, 0, 0, 0.1, 0.1, 0.1, 0.1, 0.15, 0.15, 0.15, 0.15, 0.1, 0.1, 0.1, 0.1, 0.1, 0.1), and the nonzero coefficients for subtype 3 are (0.1, 0.1, 0.1, 0.1, 0.1, 0.15, 0.15, 0.15, 0.15, 0, 0.15, 0.15, 0.15, 0.15, 0, 0, 0.1, 0.1, 0.1, 0.1). Thus, across the three subtypes, there are 60 SNPs associated with prognosis, representing 12 genes.

Two scenarios under the heterogeneity model are considered. Under the first scenario, all three subtypes share three common susceptibility genes, and each subtype has one subtype-specific susceptibility gene. The ``unmatching rate" of susceptibility genes is thus 25\%. Under the second scenario, the three subtypes share two common susceptibility genes, and each subtype has two subtype-specific susceptibility genes. The unmatching rate of susceptibility genes is 50\%. As a special case of the heterogeneity model, the homogeneity model is also considered, under which all three subtypes have the same susceptibility genes.

The logarithms of event times are generated from the AFT models with intercept equal to 0.5 and normally distributed random errors. The logarithms of censoring times are generated as uniformly distributed and independent of the event times. The censoring distribution parameters are adjusted so that overall censoring rate is about 30\%.

Beyond the proposed approach, simulated data are also analyzed using a meta-analysis approach. Here each subtype is analyzed using the group Lasso (GLasso) approach, where a ``group" corresponds to one gene with multiple SNPs. Then the identified gene lists are combined across subtypes. With both approaches, the tuning parameters are chosen using BIC. With the proposed approach, we experiment with $\gamma=$0.5, 0.7, and 0.9. Summary statistics on gene identification accuracy based on 100 replicates are shown in Table 2-3 and 5-8 (Appendix).

Simulation suggests that performance of the proposed approach depends on the correlation structure, values of nonzero regression coefficients, and $\gamma$ value. As correlation gets stronger, in general, more true positives and more false positives are identified. We fail to observe a clear pattern for the dependence (of the performance of proposed approach) on nonzero regression coefficients. As $\gamma$ gets larger, also more true positives and more false positives are identified. This observation is reasonable, considering that when $\gamma \to 1$, bridge penalization becomes close to Lasso, and that Lasso-type penalization tends to over-select. Under almost all simulated scenarios, the proposed approach identifies more true positives than GLasso. For example in Table 2, under the AR correlation with $\rho=0.5$, GLasso identifies 6.7 true positives, whereas the proposed approach identifies 10.7, 10.8, and 11.0 true positives under different $\gamma$ values. Under all simulated scenarios, the proposed approach identifies much fewer false positives. For example in Table 2, under the banded correlation scenario 2, GLasso identifies 38 false positives, whereas the proposed approach identifies 4, 4.3, and 9.8 false positives under different $\gamma$ values. Observations under the homogeneity model (Table 7 and 8) are similar. We have experimented with a few other settings and reached similar conclusions.

\section{Analysis of NHL Genetic Association Data}
NHL is the fifth leading cause of cancer incidence and mortality in the US and remains poorly understood and largely incurable. A genetic association study was conducted, searching for SNPs/genes associated with overall survival in NHL patients (Zhang et al. 2005). The prognostic cohort consists of 575 NHL patients, among whom 496 donated either blood or buccal cell samples. All cases were classified into NHL subtypes according to the World Health Organization classification system. Specifically, 155 had DLBCL, 117 had FL, 57 had CLL/SLL, 34 had MZBL, 37 had T/NK-cell lymphoma, and 96 had other subtypes. Because of sample size consideration, we focus on DLBCL, FL, and CLL/SLL, the three largest subtypes in this dataset. The study cohort was assembled in Connecticut between 1996 and 2000. Vital status of all subjects was abstracted from the CTR (Connecticut Tumor Registry) in 2008.

When genotyping, we took a candidate gene approach. Specifically, a total of 1462 tag SNPs from 210 candidate genes related to immune response were genotyped using a custom-designed GoldenGate assay. In addition, 302 SNPs in 143 candidate genes previously genotyped by Taqman assay were also included. There were a total of 1764 SNPs, representing 333 genes. Data preprocessing is conducted. Subjects with more than 20\% SNPs missing are removed from analysis; Then SNPs with more than 20\% missing are removed. The genotyping data were missing for the following reasons: the amount of DNA was too low, samples failed to amplify, samples amplified but their genotype could not be determined due to ambiguous results, or the DNA quality was poor. The remaining missing SNP measurements are then imputed. A total of 1,633 SNPs pass processing, representing 238 genes.

For DLBCL, 139 patients pass processing. Among them, 61 died, with survival times ranging from 0.47 to 10.46 years (mean 4.16 years). For the 78 censored patients, the follow up times range from 5.58 to 11.45 years (mean 9.08 years). For FL, 102 patients pass processing. Among them, 33 died, with survival times ranging from 0.91 to 10.23 years. For the 69 censored patients, the follow up times range from 4.96 to 11.39 years, with mean 8.83 years. For CLL/SLL, 50 patients pass processing. Among them, 27 died, with survival times ranging from 1.91 to 10.13 years (mean 4.85 years). For the 23 censored patients, the follow up times range from 4.92 to 11.07 years, with mean 8.83 years.

Analysis results using the proposed approach are shown in Table 4 and 9-11 (Appendix). In particular, Table 4 contains the $L_2$ norms of the identified genes, whereas Table 9-11 contain the estimated regression coefficients for SNPs. Fourteen genes are identified as associated with the overall survival of DLBCL; Twelve genes are identified as associated with FL; And five genes are identified as associated with CLL/SLL. Among the identified genes, MBP and STAT4 are shared by all three subtypes, ALOX5, IL10, IRAK2, LMAN1, MIF, and NCF4 are shared by two subtypes, and thirteen other genes are identified as subtype-specific.

Among genes shared by multiple subtypes, gene ALOX5 encodes a member of the lipoxygenase gene family and plays a dual role in the synthesis of leukotrienes from arachidonic acid. Mutations in the promoter region of this gene lead to a diminished response to antileukotriene drugs used in the treatment of asthma and are associated with atherosclerosis and several cancers. Studies that have identified this gene as a marker of NHL include Mahshid et al. (2009), Feltenmark et al. (1995) and others.
The protein encoded by gene IL10 is a cytokine produced primarily by monocytes and to a lesser extent by lymphocytes. This cytokine has pleiotropic effects in immunoregulation and inflammation. It down-regulates the expression of Th1 cytokines, MHC class II Ags, and costimulatory molecules on macrophages. It also enhances B cell survival, proliferation, and antibody production. This cytokine can block NF-kappa B activity, and is involved in the regulation of the JAK-STAT signaling pathway. The involvement of IL10 in NHL has been proposed in Bi et al. (2012) and Deng et al. (2013). IRAK2 encodes the interleukin-1 receptor-associated kinase 2, one of two putative serine/threonine kinases that become associated with the interleukin-1 receptor (IL1R) upon stimulation. It is identified as associated with NHL in Ngo et al. (2011). The protein encoded by gene LMAN1 is a type I integral membrane protein localized in the intermediate region between the endoplasmic reticulum and the Golgi. The protein is a mannose-specific lectin and a member of a novel family of plant lectin homologs in the secretory pathway of animal cells. Mutations in the gene are associated with a coagulation defect. It has been identified as a marker for gastric, colorectal, and prostate cancer. Myelin basic protein (MBP) is a protein important in the process of myelination of nerves in the central nervous system (CNS). MBP plays an important role in demyelinating diseases such as multiple sclerosis. Its involvement in NHL has been discussed in Hu et al. (2012) and Han et al. (2010).
Gene MIF encodes a lymphokine involved in cell-mediated immunity, immunoregulation, and inflammation. It plays a role in the regulation of macrophage function in host defense through the suppression of anti-inflammatory effects of glucocorticoids. This lymphokine and the JAB1 protein form a complex in the cytosol near the peripheral plasma membrane, which may indicate an additional role in integrin signaling pathways. Studies such as Xue et al. (2010) and Talos et al. (2005) have identified it as a marker of NHL. The protein encoded by gene NCF4 is a cytosolic regulatory component of the superoxide-producing phagocyte NADPH-oxidase, a multicomponent enzyme system important for host defense. It is identified as an NHL susceptibility gene in Kim et al. (2012). The protein encoded by gene STAT4 is a member of the STAT family of transcription factors. In response to cytokines and growth factors, STAT family members are phosphorylated by the receptor associated kinases, and then form homo- or heterodimers that translocate to the cell nucleus where they act as transcription activators. This protein is essential for mediating responses to IL12 in lymphocytes, and regulating the differentiation of T helper cells. Involvement of this gene in NHL risk and progression has been discussed in Chang et al. (2010) and Chang et al. (2009).

The relative stability of identified genes is evaluated using a random sampling approach (Huang and Ma 2010). In particular, we randomly sample 3/4 of the subjects and apply the proposed approach to identify prognosis-associated genes. This process is repeated 100 times. For each gene, we compute the probability of it being identified out of the 100 samplings. This probability is referred to as the observed occurrence index in Huang and Ma (2010) and measures the relative stability. Table 4 shows that only gene IL10 for DLBCL has a low occurrence index (0.21). All other observed occurrence indexes are high, suggesting relatively satisfactory stability. Prediction performance is also evaluated using a random sampling approach. In particular, genes are identified and models are constructed using 3/4 of randomly sampled subjects. Then prediction is made for the rest 1/4 subjects. Based on the predicted $X^{m\prime}\hat\beta^m$, subjects are separated into two risk groups. The logrank statistic is computed to compare the survival risk of the two groups. This process is repeated 100 times, and the mean logrank statistic is computed as 7.1 (p-value 0.0077), suggesting satisfactory prediction.

For comparison, we also analyze each subtype separately using GLasso (results shown in Table 12, Appendix). Twenty-six genes are identified as associated with the prognosis of DLBCL, seventeen genes are identified as associated with FL, and eight genes are identified as associated with CLL/SLL. Among those genes, three are shared by two subtypes. The identified genes are different from those using the proposed approach. Computation of the observed occurrence index shows that the identified genes also have satisfactory stability. Prediction evaluation generates a logrank statistic of 0.2 (p-value 0.65), which is considerably smaller than that using the proposed approach.

\section{Discussion}
In this study, with prognosis data on multiple subtypes of the same cancer, we develop a penalization approach which can conduct integrative analysis, identify important genes that contain SNPs associated with multiple subtypes, and allow for subtype-specific susceptibility genes. The proposed approach can be realized using an effective iterative algorithm. Under mild conditions, it has the much desired consistency properties. Simulation shows that the proposed approach outperforms penalization-based meta-analysis, with more true positives and fewer false positives. In the analysis of NHL prognosis data, it identifies multiple genes shared by two or three subtypes as well as subtype-specific genes. The shared genes have important biological implications. The proposed approach also leads to significantly better prediction performance.

To avoid confusion, in our description we focus on the scenario with multiple subtypes of the same cancer and the ``SNP-within-gene" structure. The proposed approach is directly applicable to the analysis of multiple types of cancers and ``gene-within-cluster (pathway)" and other structures. In addition, with minor modifications, analysis of prognosis data under other models and analysis of diagnosis data can be conducted. The proposed penalty is built on bridge-type penalties. We conjecture that it is possible to build on other penalties such as MCP. Our limited investigation shows that under the present setup, the proposed penalty may have the lowest computational cost. In data analysis, our preliminary search shows that the common genes shared by multiple subtypes have important implications. However, because of the following limitations, the analysis results should be interpreted with caution. First, the sample size is still limited. Second, the NHL study takes a candidate gene approach. It is possible that important genes have been missed in the profiling stage. Third, the proposed evaluation is cross-validation based. Although it can compare different approaches on the same ground, it does not use completely independent data. More, independent studies are needed to fully comprehend the data analysis results.

\begin{singlespace}
\section*{References}
\begin{enumerate}

\item
Bi X, Zheng T, Lan Q, Xu Z, Chen Y, Zhu G, Foss F, Kim C, Dai M, Zhao P, Holford T, Leaderer B, Boyle P, Deng Q, Chanock SJ, Rothman N, Zhang Y. Am J Hematol. (2012) Genetic polymorphisms in IL10RA and TNF modify the association between blood transfusion and risk of non-Hodgkin lymphoma. 87:766-769.

\item
Chang HC, Han L, Goswami R, Nguyen ET, Pelloso D, Robertson MJ, Kaplan MH. (2009) Impaired development of human Th1 cells in patients with deficient expression of STAT4. Blood. 113(23):5887-5890.

\item
Chang JS, Wiemels JL, Chokkalingam AP, Metayer C, Barcellos LF, Hansen HM, Aldrich MC, Guha N, Urayama KY, Scélo G, Green J, May SL, Kiley VA, Wiencke JK, Buffler PA. (2010) Genetic polymorphisms in adaptive immunity genes and childhood acute lymphoblastic leukemia. Cancer Epidemiol Biomarkers Prev. 19(9): 2152-2163.

\item
Deng Q, Zheng T, Lan Q, Lan Y, Holford T, Chen Y, Dai M, Leaderer B, Boyle P, Chanock SJ, Rothman N, Zhang Y. (2013) Occupational solvent exposure, genetic variation in immune genes, and the risk for non-Hodgkin lymphoma.
Eur J Cancer Prev. 22(1): 77-82.

\item
Feltenmark S, Runarsson G, Larsson P, Jakobsson PJ, Bjorkholm M, Claesson HE. (1995) Diverse expression of cytosolic phospholipase A2, 5-lipoxygenase and prostaglandin H synthase 2 in acute pre-B-lymphocytic leukaemia cells. Br J Haematol. 90(3): 585-594.

\item
Goh KI, CHoi IG. (2012) Exploring the human diseasome: the human disease network. Brief Funct Genomics. 11(6): 533-542.

\item
Han S, Lan Q, Park AK, Lee KM, Park SK, Ahn HS, Shin HY, Kang HJ, Koo HH, Seo JJ, Choi JE, Ahn YO, Chanock SJ, Kim H, Rothman N, Kang D. (2010) Polymorphisms in innate immunity genes and risk of childhood leukemia. Hum Immunol. 71(7): 727-730.

\item
Han X, Li Y, Huang J, Zhang Y, Holford T, Lan Q, Rothman N, Zheng T, Kosorok MR, Ma S. (2010) Identification of predictive pathways for non-Hodgkin lymphoma prognosis. Cancer Informatics. 9, 281-292.

\item
Hu W, Bassig BA, Xu J, Zheng T, Zhang Y, Berndt SI, Holford TR, Hosgood HD, Leaderer B, Yeager M, Menashe I, Boyle P, Zou K, Zhu Y, Chanock S, Lan Q, Rothman N. (2012) Polymorphisms in pattern-recognition genes in the innate immunity system and risk of non-Hodgkin lymphoma. Environ Mol Mutagen, In press.

\item
Huang J, Ma S. (2010) Variable selection in the accelerated failure time model via the bridge method. Lifetime Data Analysis. 16, 176-195.

\item
Huang J, Ma S, Xie H, Zhang C. (2009) A group bridge approach for variable selection. Biometrika, 96(2), 339-355.

\item
Huang J, Wei F, Ma S. (2012) Semiparametric regression pursuit. Statistica Sinica. 22: 1403-1426.

\item
Kim C, Zheng T, Lan Q, Chen Y, Foss F, Chen X, Holford T, Leaderer B, Boyle P, Chanock SJ, Rothman N, Zhang Y. (2012) Genetic polymorphisms in oxidative stress pathway genes and modification of BMI and risk of non-Hodgkin lymphoma. Cancer Epidemiol Biomarkers Prev. 21(5):866-868.

\item
Liu J, Huang J, Ma S. (2012) Integrative analysis of cancer diagnosis studies with composite penalization. Scandinavian Journal of Statistics. In press.

\item
Ma S, Huang J, Moran M. (2009) Identification of genes associated with multiple cancers via integrative analysis. BMC Genomics, 10: 535.

\item
Ma S, Zhang Y, Huang J, Han X, Holford T, Lan Q, Rothman N, Boyle P, Zheng T. (2010) Identification of Non-Hodgkin's lymphoma prognosis signatures using the CTGDR method. Bioinformatics. 26, 15-21.

\item
Ma S, Zhang Y, Huang J, Huang Y, Lan Q, Rothman N, Zheng T. (2012) Integrative analysis of cancer prognosis data with multiple subtypes using regularized gradient descent. Genetic Epidemiology. In press.

\item
Mahshid Y, Lisy MR, Wang X, Spanbroek R, Flygare J, Christensson B, Björkholm M, Sander B, Habenicht AJ, Claesson HE. (2009) High expression of 5-lipoxygenase in normal and malignant mantle zone B lymphocytes. BMC Immunol. 10:2.

\item
Ngo VN, Young RM, Schmitz R, Jhavar S, Xiao W, Lim KH, Kohlhammer H, Xu W, Yang Y, Zhao H, Shaffer AL, Romesser P, Wright G, Powell J, Rosenwald A, Muller-Hermelink HK, Ott G, Gascoyne RD, Connors JM, Rimsza LM, Campo E, Jaffe ES, Delabie J, Smeland EB, Fisher RI, Braziel RM, Tubbs RR, Cook JR, Weisenburger DD, Chan WC, Staudt LM. (2011)  Oncogenically active MYD88 mutations in human lymphoma. Nature. 470: 115-119.

\item
Rhodes DR, Yu J, Shanker K, Deshpande N, Varambally R, Ghosh D, Barrette T, Pandey A, Chinnaiyan AM. (2004) Large-scale meta-analysis of cancer microarray data identifies common transcriptional profiles of neoplastic transformation and progression. PNAS. 101(25): 9309-9314.

\item
Stute W. (1996) Distributional convergence under random censorship when covariables are present. Scandinavian Journal of Statistics. 23. 461-471.

\item
Talos F, Mena P, Fingerle-Rowson G, Moll U, Petrenko O. (2005) MIF loss impairs Myc-induced lymphomagenesis. Cell Death Differ. 12(10):1319-1328.

\item
Tseng P. (2001) Convergence of a block coordinate descent method for nondifferentiable minimization. J. Optimization Theory and Applications. 109: 475-494.

\item
van de Geer S. (2008) High-dimensional generalized linear models and the Lasso. Annals of Statistics. 36: 614-645.

\item
van der Vaart AW, Wellner JA. (1996) Weak Convergence and Empirical Processes: with Applications to Statistics. Springer. New York.

\item
Xue Y, Xu H, Rong L, Lu Q, Li J, Tong N, Wang M, Zhang Z, Fang Y. (2010) The MIF-173G/C polymorphism and risk of childhood acute lymphoblastic leukemia in a Chinese population. Leuk Res. 34(10):1282-1286.

\item
Yuan M, Lin Y. (2006) Model selection and estimation in regression with grouped variables. JRSSB. 68: 49-67.

\item
Zhang Y, Dai Y, Zheng T, Ma S. (2011) Risk factors of Non-Hodgkin lymphoma. Expert Opinion on Medical Diagnostics. 5; 539-550.

\item
Zhang Y, Lan Q, Rothman N,  Zhu Y, Zahm SH, Wang SS, Holford TR, Leaderer B, Boyle P, Zhang B, Zou K, Chanock S, Zheng T. (2005) A putative exonic splicing polymorphism in the BCL6 gene and the risk of non-Hodgkin lymphoma. J Natl Cancer Inst 97: 1616-1618.

\end{enumerate}
\end{singlespace}

\clearpage
\begin{table}[!tpb]
\caption{Matrix of regression coefficients for a cancer study with three subtypes, four genes and eight SNPs. An empty cell corresponds to a zero regression coefficient.}
\label{Tab:01}
\centering 
{%
\begin{tabular}{ccccc}
\hline
     &     & \multicolumn{3}{c}{Subtype}\\
Gene & SNP & S1 & S2 & S3 \\ \hline
1    & $1\_1$ & 0.20 & 0.19 & 0.21 \\
     & $1\_2$ & -0.22 &-0.19 &-0.21\\
2    & $2\_1$ & 0.18 & 0.21 & \\
     & $2\_2$ & -0.21 &-0.21 & \\
3    & $3\_1$ &      &      & 0.21\\
     & $3\_2$ &      &      & -0.18\\
4    & $4\_1$ &      &      & \\
     & $4\_2$ &      &      & \\
\hline
\end{tabular}%
}
\end{table}

\begin{table}[!tpb]
\caption{Simulation under the heterogeneity model:
unmatching rate=25$\%$ and nonzero regression coefficients under case 1. In each cell, the first row is number of true positives (standard deviation), and the second row is model size (standard deviation).}
\label{Tab:05}
\centering 
{%
\begin{tabular}{ccccc}
\hline
Correlation & \multicolumn{1}{c}{GLasso} & \multicolumn{3}{c}{Proposed} \\
\cline{2-5}
   &&\multicolumn{1}{c}{$\gamma=0.5$}&\multicolumn{1}{c}{$\gamma=0.7$}
   &\multicolumn{1}{c}{$\gamma=0.9$}\\
\hline
AR $\rho=$0.2&5.7(2.4)&6.7(5.2)&7.4(4.6)&8.9(2.8)\\
&36.4(16.5)&9.5(7.6)&10.7(6.6)&20.1(6.3)\\
AR $\rho=$0.5&6.7(2.4)&10.7(3.0)&10.8(2.8)&11.0(1.9)\\
&39.7(19.3)&14.4(4.5)&14.7(4.0)&17.1(4.5)\\
AR $\rho=$0.8&9.4(2.3)&12.0(0.2)&12.0(0.2)&11.9(0.2)\\
&50.2(19.3)&14.4(1.9)&14.6(1.9)&15.6(2.7)\\
Banded 1&5.3(2.6)&7.8(5.0)&8.6(4.4)&8.9(3.6)\\
&33.8(16.2)&11.0(7.1)&11.9(5.9)&20.5(7.0)\\
Banded 2&7.5(2.8)&10.6(3.3)&11.2(2.2)&11.3(1.7)\\
&45.5(21.7)&14.6(4.8)&15.5(3.7)&21.1(7.0)\\
Banded 3&7.9(2.4)&11.5(1.9)&11.5(1.5)&11.6(1.1)\\
&44.2(17.7)&15.4(3.0)&15.4(2.7)&18.8(5.2)\\
\hline
\end{tabular}%
}
\end{table}

\begin{table}[!tpb]
\caption{Simulation under the heterogeneity model:
unmatching rate=50$\%$ and nonzero regression coefficients under case 1. In each cell, the first row is number of true positives (standard deviation), and the second row is model size (standard deviation).}
\label{Tab:06}
\centering 
{%
\begin{tabular}{ccccc}
\hline
Correlation & \multicolumn{1}{c}{GLasso}  & \multicolumn{3}{c}{Proposed} \\
\cline{2-5}
   & &\multicolumn{1}{c}{$\gamma=0.5$}&\multicolumn{1}{c}{$\gamma=0.7$}
   &\multicolumn{1}{c}{$\gamma=0.9$}\\
\hline
AR $\rho=$0.2&4.6(2.3)&3.7(4.1)&5.8(3.8)&8.3(2.1)\\
&30.7(16.9)&5.4(7.0)&10.3(6.9)&22.2(6.8)\\
AR $\rho=$0.5&6.5(3.0)&9.6(3.5)&9.7(3.1)&10.1(2.3)\\
&36.1(20.1)&16.8(7.5)&17.2(6.9)&21.1(6.8)\\
AR $\rho=$0.8&9.7(2.2)&11.5(0.9)&11.5(0.7)&11.6(0.7)\\
&54.5(17.5)&19.4(3.9)&18.4(4.1)&20.2(5.8)\\
Banded 1&4.7(2.4)&4.4(4.1)&5.9(3.8)&7.7(2.9)\\
&34.5(17.6)&6.7(7.2)&10.8(7.7)&20.8(9.2)\\
Banded 2&7.5(2.5)&8.4(3.9)&9.1(3.1)&10.0(1.7)\\
&47.1(19.4)&14.4(7.7)&15.9(6.1)&23.5(7.1)\\
Banded 3&7.3(2.4)&9.5(2.9)&9.9(2.6)&10.1(2.2)\\
&43.2(19.6)&16.6(6.7)&17.7(5.9)&22.9(7.6)\\
\hline
\end{tabular}%
}
\end{table}

\begin{table}[!tpb]
\caption{Analysis of the NHL data using the proposed approach: $L_2$-norm of estimate for a specific gene; OOI: observed occurrence index.}
\label{cgb}
\centering 
{%
\begin{tabular}{ccccccc}
\hline
Gene & \multicolumn{2}{c}{DLBCL} & \multicolumn{2}{c}{FL} & \multicolumn{2}{c}{CLL/SLL} \\
\cline{2-7}
   & $L_2$-norm & OOI & $L_2$-norm & OOI & $L_2$-norm &OOI\\
\hline
ALOX12	&		&		&	0.02	&	0.83	&		&		\\
ALOX15B	&	0.01	&	0.76	&		&		&		&		\\
ALOX5	&	0.02	&	0.71	&		&		&	0.01	&	0.64	\\
CLCA1	&		&		&	0.02	&	0.83	&		&		\\
CSF2	&	0.02	&	0.86	&		&		&		&		\\
DEFB1	&	0.03	&	0.97	&		&		&		&		\\
IL10	&	1.E-04	&	0.21	&		&		&	0.02	&	0.62	\\
IL17C	&		&		&	0.02	&	0.78	&		&		\\
IRAK2	&		&		&	0.02	&	0.88	&	0.01	&	0.80	\\
LIG4	&		&		&	0.01	&	0.87	&		&		\\
LMAN1	&	0.02	&	0.71	&	0.02	&	0.77	&		&		\\
MBP	&	0.01	&	0.68	&	4.E-03	&	0.60	&	0.04	&	0.68	\\
MCP	&	0.01	&	0.83	&		&		&		&		\\
MEFV	&	0.02	&	0.85	&		&		&		&		\\
MIF	&	0.02	&	0.83	&	1.E-03	&	0.53	&		&		\\
MUC6	&	0.03	&	0.99	&		&		&		&		\\
NCF4	&	0.01	&	0.64	&	0.01	&	0.64	&		&		\\
PTK9L	&		&		&	0.01	&	0.66	&		&		\\
SERPINB3	&	0.01	&	0.55	&		&		&		&		\\
SOD3	&		&		&	4.E-03	&	0.47	&		&		\\
STAT4	&	0.02	&	0.95	&	0.01	&	0.88	&	0.01	&	0.91	\\
\hline
\end{tabular}%
}
\end{table}

\clearpage
\setcounter{page}{1}
\section{Appendix}

\subsection{Statistical properties and proofs}
In this section, we use the same notations as in Section 2 and 3.1. The PCA step described in Section 3.4 is optional and mainly for practical consideration. If PCA is actually conducted, the $d_{jk}$ values may be smaller.

\begin{proof}[Proof of Proposition~\ref{prop1}]
We have $\min_{\beta,\theta}S(\beta,\theta)=\min_{\beta}\hat{S}(\beta)$,\\ where $\hat{S}(\beta)=\min_{\theta}\{S(\beta,\theta):\theta\ge 0\}$. For any $\beta$,
\begin{equation}
 \hat\theta(\beta)\equiv\argmin\{S(\beta,\theta):\theta\ge 0\}=\argmin\left\{\sum_{j=1}^J \theta_j^{1-1/\gamma}c_j^{1/\gamma}\left(\sum_{k=1}^{M_j} \sqrt{d_{jk}}\left \Vert\beta_{jk}\right \Vert\right) +\tau_n\sum_{j=1}^J\theta_j\right\}. \notag
\end{equation}
Therefore, $\hat\theta(\beta)=(\hat\theta_1(\beta),\dots,\hat\theta_J(\beta))^{\prime}$ satisfies
\begin{equation}
  (1/\gamma-1)\hat\theta_j^{-1/\gamma}(\beta) c_j^{1/\gamma}\sum_{k=1}^{M_j} \sqrt{d_{jk}}\left \Vert\beta_{jk}\right \Vert=\tau_n,\quad(j=1,\dots,J). \notag
\end{equation}
Write $\hat{S}(\beta)=S(\beta,\hat\theta(\beta))$, and substitute the expressions
\begin{equation}
\hat\theta_j(\beta)=c_j\left(\frac{\gamma}{1-\gamma} \right)^{1-\gamma}\left(\sum_{k=1}^{M_j} \sqrt{d_{jk}}\left \Vert\beta_{jk}\right \Vert \right)^{\gamma},\quad \hat\theta_j^{1-1/\gamma}(\beta)=\left( \frac{\gamma}{1-\gamma}\right)\frac{c_j^{1-1/\gamma}
\tau_n^{1-\gamma}}{\left(\sum_{k=1}^{M_j} \sqrt{d_{jk}}\left \Vert\beta_{jk}\right \Vert \right)^{1-\gamma}} \notag
\end{equation}
into $S(\beta,\hat\theta(\beta))$. After some algebra, we obtain
\begin{equation}
\hat{S}(\beta)=\frac{1}{2n}\left \Vert Y-X\beta \right \Vert^2+\lambda_n\sum_{j=1}^Jc_j\left(\sum_{k=1}^{M_j} \sqrt{d_{jk}}\left \Vert\beta_{jk}\right \Vert \right)^{\gamma}. \notag
\end{equation}
Then if $\lambda_n = \tau_n^{1-\gamma}\gamma^{-\gamma}(1-\gamma)^{\gamma-1}$,
$\hat{S}(\beta)$ and the objective function defined in Section 3.1 are equivalent.
\end{proof}


Let $B_1$ and $B_2$ be the index sets of genes with nonzero and zero norms of regression coefficients, respectively. 
Let $\beta_0$ be the true parameter vector of $\beta$, $q^m$ be the size of the set $\{j:\lVert\beta_{0jm}\rVert\ne 0\}$, and $q=\sum_{m=1}^M q^m$ (that is, the total number of associations between genes and prognosis across all subtypes). Without loss of generality, suppose that $\lVert\beta_{0jk}\rVert\ne 0$ for all $j=1,\dots,J_1$ and some $k=1,\dots, M_j$. 
With $X$ as the full design matrix (defined in Section 3.2), denote $X_{B1}$ as the design matrix corresponding to $B_1$.
Define $\Sigma_n =n^{-1}X^{\prime}X,~\Sigma_{1n}=n^{-1}X_{B1}^{\prime} X_{B1}$. Further for any index set $A$, let $X_A$ denote the corresponding design matrix, and $\Sigma_A=n^{-1}X_A^{\prime}X_A$. Let $A_1=\{(j,k):\lVert \hat\beta_{0jk}\rVert\ne 0\}$.

We make the following assumptions:
(A1) $M$ is finite. $q$ is finite. For subtype $m~(=1,\ldots, M)$, $\{(Y_{oi}^m,\delta_i^m, X_{oi}^m), i=1,\ldots, n^m \}$ are iid. The errors $\{\epsilon_1^m, \ldots, \epsilon_{n^m}^m\}$ are iid with mean zero and finite variance $(\sigma^{m})^2$. Denote $\epsilon^m=(\epsilon_1^m, \ldots, \epsilon_{n^m}^m)^{\prime}$, and $\epsilon=(\epsilon^{1\prime}, \ldots, \epsilon^{M\prime})^{\prime}$.
Denote $\sigma^2 = \max_m (\sigma^{m})^2$.
The $i$th component of $\epsilon^m$, $\epsilon_i^m$, is subgaussian,
in the sense that there exist $K_1,K_2>0$ such that the tail probability satisfies
$P(|\epsilon_i^m|>u)\le K_2\mathrm{exp}(-K_1 u^2)$ for all $u\ge0$. (A2) For subtype $m~(=1,\ldots, M)$, the errors $(\epsilon_{1}^m, \ldots, \epsilon_{n^m}^m)$ are independent of the Kaplan-Meier weights $(\omega_1^m, \ldots, \omega_{n^m}^m)$. The covariates are bounded. 
(A3) 
The design matrix satisfies the sparse Riesz condition (SRC) with rank $q^*$. That is, there exist constants $0<c_*<c^*<\infty$, such that for $q^*=(3+4C)q$ and $C=c^*/c_*$, with probability converging to 1, $c_*\le \frac{\nu^{\prime}\Sigma_A \nu}{\lVert \nu \rVert^2} \leq c^*$, $\forall A$ with $\lvert A \rvert =q^*$ and $\nu\in \mathbb{R}^{q^*}$. In addition, the SRC condition holds for the design matrix of each subtype separately with rank $q^{m*}$ for subtype $m$.
(A4) $\lambda_n(\mathrm{log}d/n)^{\gamma/2-1}\rightarrow \infty$. Let $\eta_n=\max_j c_j \left(\sum_{k=1}^{M_j}\sqrt{d_{jk}}\lVert \beta_{0jk}\rVert\right)^{\gamma-1}$. $\sum_{j=1}^J \sum_{k=1}^{M_j}\sqrt{d_{jk}}\lVert \beta_{0jk}\rVert=O(1)$ and $\eta_n=O(1)$.

The above assumptions are in parallel with those in Huang and Ma (2010). Further complications are introduced to accommodate the multi-datatsets setting, heterogeneity across subtypes, and ``SNP-within-gene" structure. Main properties of the penalized estimate can be summarized as follows.

\begin{thm}\label{theorem1}
Suppose the (A1)--(A3) hold and $\lambda_n\ge O(1)\sqrt{\mathrm{log(d)/n}}$. Then
\begin{enumerate}
  \item With probability converging to 1, $\lvert A_1\rvert\le (2+4C)q$.
  \item  $\lVert\hat\beta-\beta_0 \rVert^2\le \frac{256\lambda_n\eta_n^2}{c_*^2} + O_p\left(\frac{\lvert A_1\rvert\mathrm{log}d}{nc_*^2}\right)$. In particular, if $\lambda_n= O(1)\sqrt{\mathrm{log(d)/n}}$, then $\lVert\hat\beta-\beta_0 \rVert^2=O_p\left(\mathrm{log}d/n\right)$.
\end{enumerate}
\end{thm}

The above result establishes that the number of identified genes is a finite multiply of the true number of associated genes, which is assumed to be finite in (A1). In addition, if $\log(d)/n \to 0$, the estimate is $L_2$ estimation consistent. The selection and estimation results can be further strengthened as follows.

\begin{thm}\label{theorem2}
Suppose that (A1)--(A4) hold.
\begin{enumerate}
\item It holds that $\hat\beta_{B_2}$=0 with probability converging to 1.
\item
Suppose that $\{B_1,\beta_{0B_1}\}$ are fixed unknown. In the AFT models, the subgaussian conditions are strengthened to normal distributions.
$\Sigma_{\omega1n^m}^m=(n^m)^{-1/2}(X_{\omega B_1^m}^m)^{\prime}X_{\omega B_1^m}^m\rightarrow \Sigma_{\omega1}^m$ and $n^{-1/2}X_{B_1}^{\prime}\epsilon=\sum_{m=1}^M (n^m)^{-1/2}(X_{\omega B_1^m}^m)^{\prime}\epsilon^m \rightarrow Z \sim N(0,\sum_{m=1}^M (\sigma^m)^2\Sigma_{\omega1}^m)$ in distribution.

Then, in distribution,
\begin{equation}
  \sqrt{n}(\hat\beta_{B_1} - \beta_{0B_1})\rightarrow \argmin \{V_1{(u)}:u\in R^{|B_1|}\}, \notag
\end{equation}
where
\begin{eqnarray}
V_1(u)&=&-2u^{\prime}Z+u^{\prime}\Sigma_1u+ \gamma\lambda_0\sum_{j=1}^{|B_1|} c_j\left(\sum_{k=1}^{M_j}\sqrt{d_{jk}}\Vert\beta_{0jk}\Vert \right)^{\gamma-1} \sum_{k=1}^{M_j} \{ u_{jk}\frac{\beta_{0jk}}{\Vert\beta_{0jk}\Vert}
\mathrm{I}(\Vert\beta_{0jk}\ne0\Vert)\notag \\
  & & +\Vert u_{jk}\Vert I(\Vert\beta_{0jk}=0)\Vert\}, \notag
\end{eqnarray}
with $u_{jk}$ corresponding to the component of gene $j$ and subtype $k$.
\end{enumerate}
\end{thm}

The above result establishes that under mild conditions, the proposed approach can consistently identify genes that are associated with at least one subtype. This result is consistent with that in Huang et al. (2009) for the analysis of a single dataset. In addition, when the random errors have normal distributions, the asymptotic distribution can be rigorously established. To prove the above theorems, we first establish the following lemma.

Let $\zeta^m=(\zeta_1^m,\dots,\zeta_n^m)^{\prime}$, where $\zeta_i^m=\omega_i^m \epsilon_{i}^m\equiv \omega_i^m (Y_{o(i)}^m - X_{o(i)}^{m\prime}\beta_0^m)$.

\begin{lem}\label{lemma0}
Suppose that assumption (A2) and (A3) hold. Let $\xi_j^m=X_{j}^{m\prime}\zeta^m,1\le j\le d$, where $X_j^m$ is the $j$th column of $X^m$.
Let $\xi_{n}^m=\mathrm{max}_{1\le j\le d} |\xi_j^m|$. Then.
 \begin{equation*}
   \mathrm{E} (\xi_n^m)\le C_1\sqrt{\mathrm{log}(d)}\left( \sqrt{2C_2n^m\mathrm{log}(d)}+4\mathrm{log}(2d)+C_2n^m\right),
 \end{equation*}
 where $C_1$ and $C_2$ are two positive constants. In particular, when $\mathrm{log}(d)/n^m \rightarrow 0$,
 \begin{equation*}
  \mathrm{E} (\xi_n^m)=O(1)\sqrt{n^m\mathrm{log}(d)}.
 \end{equation*}
\end{lem}

\begin{proof}[Proof of Lemma~\ref{lemma0}]
Let $s_{n^mj}^2=X_j^{m\prime}X_j^m$. Conditional on $X_j^m$'s, assumption (A2) and (A3) imply that $\xi_j^m$'s are subgaussian. Let $s_{n^m}^2=\mathrm{max}_{1\le j\le d} s_{n^mj}^2$. By (A2) and the maximal inequality for subgaussian random variables (Van der Varrt and Wellner 1996, Lemmas 2.2.1 and 2.2.2),
\begin{equation*}
 \mathrm{E}\left(\quad\underaccent{1\le j\le d}{\mathrm{max}}\quad |\xi_j^m| \Big| X_j^m,1\le j\le d \right)\le C_1s_{n^m}\sqrt{\mathrm{log}(d)},
\end{equation*}
for a constant $C_1>0$. Therefore,
\begin{equation}
 \label{ineq0}
  \mathrm{E}\left(\quad\underaccent{1\le j\le d}{\mathrm{max}}\quad|\xi_j^m|\right)\le C_1\sqrt{\mathrm{log}(d)}\mathrm{E}(s_{n^m}).
\end{equation}
Since
\begin{equation}
\label{ineq1}
\sum_{i=1}^{n^m}\mathrm{E}\left( X_{j(i)}^{m2}- \mathrm{E}X_{j(i)}^{m2} \right)^2 \le 4C_2n^m,
\end{equation}
and
\begin{equation}
\label{ineq2}
\underaccent{1\le j\le d}{\mathrm{max}}\sum_{i=1}^{n^m} \mathrm{E}X_{j(i)}^{m2} \le C_2n^m.
\end{equation}
By Lemma 4.2 of Van der Geer (2008), (\ref{ineq1}) implies
\begin{equation}
  \mathrm{E}\left( \underaccent{1\le j\le d}{\mathrm{max}}\left\vert \sum_{i=1}^2(X_{j(i)}^{m2}-\mathrm{E}X_{j(i)}^{m2})\right\vert\right)\le \sqrt{2C_2n^m\mathrm{log}(d)}+4\mathrm{log}(2d).
\end{equation}
Therefore, by (\ref{ineq2}) and the triangle inequality,
\begin{equation*}
  \mathrm{E} s_{n^m}^2 \le \sqrt{2C_2n^m\mathrm{log}(d)}+4\mathrm{log}(2d)+C_2n^m.
\end{equation*}
Now since $\mathrm{E}s_{n^m}\le (\mathrm{E}s_{n^m}^2)^{1/2}$, we have
\begin{equation}
\label{ineq3}
 \mathrm{E} s_{n^m} \le \left( \sqrt{2C_2n^m\mathrm{log}(d)}+4\mathrm{log}(2d)+C_2n^m\right)^{1/2}.
\end{equation}
The lemma follows from (\ref{ineq0}) and (\ref{ineq3}).
\end{proof}


\begin{proof}[Proof of Theorem~\ref{theorem1}]
Part (i) follows from the proof of Theorem 1 of Zhang and Huang (2008). One difference is that here a subgaussian assumption is made, which is weaker than the normality assumption in Zhang and Huang (2008). Since subgaussian random variables have the same tail behaviors as normal random variables, the argument of Zhang and Huang (2008) goes through. In addition, the SRC condition is imposed at a group level, as opposed to individual covariate level, to accommodate the ``SNP-within-gene" structure.

(ii) By the definition of $\hat\beta$,
\begin{equation*}
\frac{1}{2n}\lVert Y-X\hat\beta\rVert^2 +\lambda_n\sum_{j=1}^J c_j \left(\sum_{k=1}^{M_j}\sqrt{d_{jk}}\lVert \hat\beta_{jk}\rVert\right)^{\gamma} \le  \frac{1}{2n}\lVert Y-X\beta_0\rVert^2 +\lambda_n\sum_{j=1}^J c_j \left(\sum_{k=1}^{M_j}\sqrt{d_{jk}}\lVert \beta_{0jk}\rVert\right)^{\gamma}.
\end{equation*}
Thus
\begin{equation*}
\frac{1}{2n}\lVert Y-X\hat\beta\rVert^2 +\lambda_n\sum_{j=1}^{J_1} c_j \left(\sum_{k=1}^{M_j}\sqrt{d_{jk}}\lVert \hat\beta_{jk}\rVert\right)^{\gamma} \le  \frac{1}{2n}\lVert Y-X\beta_0\rVert^2 +\lambda_n\sum_{j=1}^{J_1} c_j \left(\sum_{k=1}^{M_j}\sqrt{d_{jk}}\lVert \beta_{0jk}\rVert\right)^{\gamma}.
\end{equation*}
Using $\zeta=Y-X\beta_0$, we have
\begin{equation}
\frac{1}{2n}\lVert Y-X\hat\beta\rVert^2-\frac{1}{2n}\lVert Y-X\beta_0\rVert^2 =\frac{1}{2n}\lVert X(\hat\beta-\beta_0)\rVert^2-\frac{1}{n}\zeta^{\prime}X(\hat\beta-\beta_0). \label{eqn1}
\end{equation}
Let $B=B_1\cup A_1=\{(j,k):\lVert\beta_{0jk}\rVert\ne 0\quad \mbox{or}\quad \lVert\hat\beta_{jk}\rVert\ne 0\}$. Note that $\lvert B\rvert\le q^*$ with probability converging to 1 by part (i), where $q^*$ is defined in (A3). Denote $\eta_B=X_B(\hat\beta_B-\beta_{0B})$.

Since $b^{\gamma}-a^{\gamma}\le 2(b-a)b^{\gamma-1}$ for $0\le a\le b$,
\begin{eqnarray}
\label{ineq4}
& & \lambda_n\sum_{j=1}^{J_1} c_j (\sum_{k=1}^{M_j}\sqrt{d_{jk}}\lVert \beta_{0jk}\rVert)^{\gamma} -\lambda_n\sum_{j=1}^{J_1} c_j (\sum_{k=1}^{M_j}\sqrt{d_{jk}}\lVert \hat\beta_{jk}\rVert)^{\gamma} \notag\\
&\le &2\lambda_n\sum_{j=1}^{J_1} \left(c_j (\sum_{k=1}^{M_j}\sqrt{d_{jk}}\lVert \beta_{0jk}\rVert)^{\gamma-1}\sum_{k=1}^{M_j}\sqrt{d_{jk}}\lvert\lVert \beta_{0jk}\rVert-\lVert \hat\beta_{jk}\rVert\rvert\right) \notag \\
& \le &2\lambda_n\eta_n\sum_{j=1}^{J_1}\sum_{k=1}^{M_j}\sqrt{d_{jk}}\lvert\lVert \beta_{0jk}\rVert-\lVert \hat\beta_{jk}\rVert\rvert \notag\\
&\le & 2\lambda_n\eta_n\lVert \hat\beta_B-\beta_{0B}\rVert.
\end{eqnarray}
Following from (\ref{eqn1}) and (\ref{ineq4}), we have
\begin{equation}
\label{ineq5}
  \frac{1}{2n}\lVert \eta_B\rVert^2-\frac{1}{n}\zeta^{\prime}\eta_B\le 2\lambda_n\eta_n\lVert \hat\beta_B-\beta_{0B}\rVert.
\end{equation}
Let $\zeta_B$ be the projection of $\zeta$ to the span of $X_B$, i.e., $\zeta_B=X_B(X_B^{\prime}X_B)^{-1}X_B^{\prime}\zeta$. We have
\begin{equation*}
  \zeta^{\prime}\eta_B=\zeta^{\prime}X_B(\hat\beta_B-\beta_{0B})=\left\{(X_B^{\prime}X_B)^{-1/2}X_B^{\prime}\zeta \right\}^{\prime}\left\{(X_B^{\prime}X_B)^{1/2}(\hat\beta_B-\beta_{0B}) \right\}.
\end{equation*}
Therefore, by the Cauchy-Schwarz inequality,
\begin{equation}
\label{ineq6}
  \lvert \zeta^{\prime}\eta_B\rvert\le\lVert\zeta_B\rVert\cdot\lVert\eta_B\rVert\le\lVert\zeta_B\rVert^2+\frac{1}{4}\lVert\eta_B\rVert.
\end{equation}
Combining (\ref{ineq5}) and (\ref{ineq6}),
\begin{equation}
\label{ineq7}
  \frac{1}{2n}\lVert \eta_B\rVert^2\le \frac{2}{n}\lVert \zeta_B\rVert^2+4\lambda_n\zeta_n\lVert \hat\beta_B-\beta_{0B}\rVert.
\end{equation}
By the SRC condition (A3),$n^{-1}\lVert \eta_B\rVert^2\ge c_*\lVert \hat\beta_B-\beta_{0B}\rVert^2$. Thus (\ref{ineq7}) implies
\begin{equation*}
  \frac{c_*}{2}\lVert \hat\beta_B-\beta_{0B}\rVert^2\le \frac{2}{n}\lVert \zeta_B\rVert^2 + \frac{4*(4\lambda_n\eta_n)^2}{c_*}+\frac{c_*}{4}\lVert \hat\beta_B-\beta_{0B}\rVert^2.
\end{equation*}
It follows that
\begin{equation}
\label{ineq8}
  \lVert \hat\beta_B-\beta_{0B}\rVert^2\le\frac{8\lVert \zeta_B\rVert^2}{nc_*}+\frac{(16\lambda_n\eta_n)^2}{c_*^2}.
\end{equation}
Now
\begin{eqnarray*}
  \lVert \zeta_B\rVert^2&=&\lVert(X_B^{\prime}X_B)^{-1/2}X_B^{\prime}\zeta\rVert^2\le \frac{1}{nc_*} \lVert X_B^{\prime}\zeta\rVert^2 \\
  & \le & \frac{1}{nc_*}\quad\quad\underaccent{A:|A|\le q^*}{\mathrm{max}}\quad\lVert X_A^{\prime}\zeta\rVert^2 \le \frac{1}{c_*}  \quad\underaccent{m}{\mathrm{max}}\quad\quad\underaccent{A^m:|A^m|\le q^{m*}}{\mathrm{max}}\quad\quad\frac{1}{n^m}\lVert X_{\omega A^m}^{m\prime}\zeta^m\rVert^2.
\end{eqnarray*}
We have
\begin{equation*}
  \underaccent{A^m:|A^m|\le q^{m*}}{\mathrm{max}}\quad\quad\lVert X_{\omega A^m}^{m\prime}\zeta^m\rVert^2=\quad\quad\underaccent{A^m:|A^m|\le q^{m*}}{\mathrm{max}}\quad\quad\sum_{j\in A^m} \lvert X_{\omega j}^{m\prime}\zeta^m\rvert^2 \le q^{m*}\quad\underaccent{1\le j\le d}{\mathrm{max}}\quad\lvert X_{\omega j}^{m\prime}\zeta^m\rvert^2.
\end{equation*}
By Lemma \ref{lemma0},
\begin{equation*}
  \underaccent{1\le j\le d}{\mathrm{max}}\quad\lvert X_{\omega j}^{m\prime}\zeta^m\rvert^2=n^m \quad\underaccent{1\le j\le d}{\mathrm{max}}\quad\lvert (n^m)^{-\frac{1}{2}}X_{\omega j}^{m\prime}\zeta^m\rvert^2=O_p(n^m\mathrm{log}d).
\end{equation*}
Therefore,
\begin{equation}
  \lVert \zeta_B\rVert^2=O_p(q^*\mathrm{log}d). \label{ineq9}
\end{equation}
The result follows from (\ref{ineq8}) and (\ref{ineq9}).
\end{proof}

\begin{proof}[Proof of Theorem~\ref{theorem2}]
(i) We prove the first part of Theorem~\ref{theorem2}.

Define $\tilde\beta$ by
\begin{equation*}
  \tilde\beta_j =
\begin{cases}
 \hat\beta_j, & (j \in B_1) \\
 0, & (j \in B_2)
\end{cases}.
\end{equation*}
The Karush-Kuhn-Tucker condition for~(\ref{update2}) implies that
\begin{eqnarray*}
  \frac{1}{n}(Y-X\hat\beta)^{\prime}X_{jl} = \gamma c_j\lambda_n \left( \sum_{k=1}^{M_j} \sqrt{d_{jk}} \lVert \hat\beta_{jk} \rVert \right)^{\gamma-1} \sqrt{d_{jl}}\frac{\hat\beta_{jl}}{\lVert \hat\beta_{jl} \rVert} I(\lVert \hat\beta_{jl}\rVert\ne 0),
\end{eqnarray*}
where $X_{jl}$ is the sub-matrix corresponding to the $j$th gene of the $l$th subtype.

Since $(\hat\beta_{jl}-\tilde\beta_{jl})^{\prime}\frac{\hat\beta_{jl}}{\lVert \hat\beta_{jl} \rVert}=\lVert \hat\beta_{jl} \rVert I(\lVert\beta_{0jl}\rVert=0)$, we have
\begin{eqnarray*}
  \frac{1}{n}(Y-X\hat\beta)^{\prime}X(\hat\beta-\tilde\beta)&=& \sum_{\{j:\lVert\beta_{0jl}\rVert=0\}\ni j}\gamma c_j \lambda_n \left( \sum_{k=1}^{M_j} \sqrt{d_{jk}} \lVert \hat\beta_{jk} \rVert \right)^{\gamma-1} \sqrt{d_{jl}} \lVert \hat\beta_{jl} \rVert \\
   &=& \sum_{\{j:\lVert\beta_{0jl}\rVert=0\}\ni j} \gamma c_j \lambda_n \left( \sum_{k=1}^{M_j} \sqrt{d_{jk}} \lVert \hat\beta_{jk} \rVert \right)^{\gamma-1} \sqrt{d_{jl}}\left( \lVert \hat\beta_{jl} \rVert - \lVert \tilde\beta_{jl} \rVert \right) \\
   &=& \gamma\lambda_n\sum_{j=1}^J  c_j \left( \sum_{k=1}^{M_j} \sqrt{d_{jk}} \lVert \hat\beta_{jk} \rVert \right)^{\gamma-1}\cdot\left\{ \left( \sum_{l=1}^{M_j}\sqrt{d_{jl}} \lVert \hat\beta_{jl} \rVert \right) - \left( \sum_{l=1}^{M_j}\sqrt{d_{jl}} \lVert \tilde\beta_{jl} \rVert \right) \right\}.
\end{eqnarray*}

Since $\gamma b^{\gamma-1}(b-a)\le b^{\gamma} - a^{\gamma}$ for $0\le a\le b$, for $j\le J_1$,
we have
\begin{eqnarray*}
& \gamma\lambda_n\sum_{j=1}^{J_1} c_j \left( \sum_{k=1}^{M_j} \sqrt{d_{jk}} \lVert \hat\beta_{jk} \rVert \right)^{\gamma-1} \left\{ \left( \sum_{l=1}^{M_j}\sqrt{d_{jl}} \lVert \hat\beta_{jl} \rVert \right) - \left( \sum_{l=1}^{M_j}\sqrt{d_{jl}} \lVert \tilde\beta_{jl} \rVert \right) \right\} \\
& \le \lambda_n \sum_{j=1}^{J_1} c_j \left\{ \left( \sum_{l=1}^{M_j}\sqrt{d_{jl}} \lVert \hat\beta_{jl} \rVert \right)^{\gamma} - \left( \sum_{l=1}^{M_j}\sqrt{d_{jl}} \lVert \tilde\beta_{jl} \rVert \right)^{\gamma} \right\}.
\end{eqnarray*}

Since $\left( \sum_{l=1}^{M_j}\sqrt{d_{jl}} \lVert \hat\beta_{jl} \rVert \right)=0$ for $j>J_1$, this implies that
\begin{eqnarray}
\label{A1}
\frac{1}{n} |(Y-X\hat\beta)^{\prime}X(\hat\beta-\tilde\beta)| &\le &  \lambda_n \sum_{j=1}^{J_1} c_j \left\{ \left( \sum_{l=1}^{M_j}\sqrt{d_{jl}} \lVert \hat\beta_{jl} \rVert \right)^{\gamma} - \left( \sum_{l=1}^{M_j}\sqrt{d_{jl}} \lVert \tilde\beta_{jl} \rVert \right)^{\gamma} \right\} \notag\\
 & &+\gamma\lambda_n \sum_{j=J_1+1}^J c_j \left( \sum_{k=1}^{M_j} \sqrt{d_{jk}} \lVert \hat\beta_{jk} \rVert \right)^{\gamma}.
\end{eqnarray}

By the definition of $\hat\beta$, we have
\begin{equation}
\frac{1}{2n}\lVert Y-X\hat\beta\rVert^2 + \lambda_n\sum_{j=1}^J c_j \left(\sum_{k=1}^{M_j}\sqrt{d_{jk}} \lVert \hat\beta_{jk}\rVert\right)^{\gamma} \le  \frac{1}{2n}\lVert Y-X\tilde\beta\rVert^2 +\lambda_n\sum_{j=1}^J c_j \left(\sum_{k=1}^{M_j}\sqrt{d_{jk}} \lVert \tilde\beta_{jk}\rVert\right)^{\gamma}. \notag
\end{equation}
Since $\left( \sum_{k=1}^{M_j}\sqrt{d_{jk}} \lVert \tilde\beta_{jk} \rVert \right)=0$ for $j>J_1$, by~(\ref{A1}),
\begin{eqnarray*}
 \frac{1}{n} |(Y-X\hat\beta)^{\prime}X(\hat\beta-\tilde\beta)| & + &(1-\gamma)\lambda_n\sum_{j=J_1+1}^J c_j \left(\sum_{k=1}^{M_j}\sqrt{d_{jk}} \lVert \hat\beta_{jk}\rVert\right)^{\gamma} \\
  &\le & \lambda_n \sum_{j=1}^J  c_j \left(\sum_{k=1}^{M_j}\sqrt{d_{jk}} \lVert \hat\beta_{jk}\rVert\right)^{\gamma} - \lambda_n \sum_{j=1}^J c_j \left(\sum_{k=1}^{M_j}\sqrt{d_{jk}} \lVert  \tilde\beta_{jk}\rVert\right)^{\gamma}  \\
  &\le & \frac{1}{2n}\lVert Y-X\tilde\beta\rVert^2 - \frac{1}{2n}\lVert Y-X\hat\beta\rVert^2 \\
  & = & \frac{1}{2n}\lVert X(\hat\beta -\tilde\beta) \rVert^2 - \frac{1}{n} (Y-X\hat\beta)^{\prime}X(\hat\beta-\tilde\beta).
\end{eqnarray*}
Thus, by SRC condition (A3), with $n^{-1}\lVert X(\hat\beta-\tilde\beta)\rVert^2\le c^*\lVert\hat\beta-\tilde\beta\rVert^2$, we have
\begin{equation*}
 (1-\gamma)\lambda_n \sum_{j=J_1+1}^J c_j \left(\sum_{k=1}^{M_j}\sqrt{d_{jk}} \lVert \hat\beta_{jk}\rVert\right)^{\gamma} \le \frac{1}{2n}\lVert X(\hat\beta-\tilde\beta)\rVert^2\le \frac{c^*}{2} \lVert\hat\beta_{B_2}\rVert^2\le \frac{c^*}{2} \lVert\hat\beta-\beta_0\rVert^2,
\end{equation*}
which implies, by Theorem~\ref{theorem1} (ii),
\begin{equation}
\label{A2}
(1-\gamma)\lambda_n \sum_{j=J_1+1}^J c_j \left(\sum_{k=1}^{M_j}\sqrt{d_{jk}} \lVert \hat\beta_{jk}\rVert\right)^{\gamma} \le \frac{c^*}{2} \lVert\hat\beta_{B_2}\rVert^2\le \frac{c^*}{2} \lVert\hat\beta-\beta_0\rVert^2 = O_p (\frac{\mathrm{log}d}{n}).
\end{equation}
We still need to find a lower bound of $\sum_{j=J_1+1}^J c_j \left(\sum_{k=1}^{M_j}\sqrt{d_{jk}} \lVert \hat\beta_{jk}\rVert\right)^{\gamma}$. Since $c_j \ge 1$ by assumption,
\begin{eqnarray}
\label{A3}
\sum_{\{j:\lVert\beta_{0jl}\rVert=0\}\ni j} c_j \left(\sum_{\{k:\lVert\beta_{0jk}\rVert=0\}\ni k}\sqrt{d_{jk}} \lVert \hat\beta_{jk}\rVert\right)^{\gamma} 
\ge \left(\sum_{(j,k):\lVert\beta_{0jk}\rVert=0}\sqrt{d_{jk}} \lVert \hat\beta_{jk}\rVert\right)^{\gamma}\ge \lVert\hat\beta_{B_2}\rVert^{\gamma}.
\end{eqnarray}
If $\lVert\hat\beta_{B_2}\rVert>0$, the combination of~(\ref{A2}) and~(\ref{A3}) yields
\begin{equation*}
 (1-\gamma)\lambda_n C_1 \le \frac{c^*}{2} \lVert\hat\beta_{B_2}\rVert^{2-\gamma}\le O_p\left(\left(\frac{\mathrm{log}d}{n}\right)^{\frac{2-\gamma}{2}} \right).
\end{equation*}
Since $\lambda_n(\mathrm{log}d/n)^{\gamma/2-1}\rightarrow \infty$ by assumption, this implies that
\begin{equation*}
  Pr(\lVert\hat\beta_{B_2}\rVert >0) \le Pr\left\{\lambda_n\left(\frac{\mathrm{log}d}{n}\right)^{\frac{\gamma}{2}-1} \le O_p(1)\right\}\rightarrow 0.
\end{equation*}

(ii) We now prove the second part. Let $h_n=n^{-1/2}$ define
\begin{equation*}
  V_{1n}(\bu)=L_n(\beta_0+h_n(\bu^{\prime},\mathbf{0}^{\prime})^{\prime})-L_n(\beta_0),
\end{equation*}
with $\mathbf{0}$ being the zero vector of dimension $|B_2|$. By (i), the following holds with large probability:
\begin{equation*}
 \hat\beta-\beta_0 =h_n(\hat\bu^{\prime},\mathbf{0}^{\prime})^{\prime}, \quad \hat\bu=\argmin \left\{ V_{1n}(\bu);\bu\in\mathbb{R}^{d_1}\right\}.
\end{equation*}
The function $V_{1n}(\bu),\bu\in\mathbb{R}^{d_1}$, can be written as
\begin{eqnarray*}
V_{1n}(\bu)&=&\left\{ -2h_n\bu^{\prime}X_1^{\prime}\epsilon + h_n^2\bu^{\prime}X_1^{\prime}X_1\bu \right\} +\lambda_n\sum_{j=1}^{J_1} c_j\left( \sum_{k=1}^{M_j} \sqrt{d_{jk}} \lVert \hat\beta_{jk} + h_n\bu_{jk}\rVert \right )^{\gamma}\\
& &-c_j\left( \sum_{k=1}^{M_j} \sqrt{d_{jk}} \lVert \hat\beta_{jk}\rVert \right )^{\gamma}\\
&=& T_{1n}(\bu)+T_{2n}(\bu).
\end{eqnarray*}
For the first term, we have
\begin{equation*}
 T_{1n}(\bu) \rightarrow_D -2\bu^{\prime}W+\bu^{\prime}\Sigma_1\bu.
\end{equation*}
For the second term,
\begin{equation*}
 T_{2n}(\bu) \rightarrow\gamma\lambda_0\sum_{j=1}^{J_1} c_j\left(\sum_{k=1}^{M_j}\sqrt{d_{jk}}\Vert\beta_{0jk}\Vert \right)^{\gamma-1} \sum_{k=1}^{M_j} \left\{ u_{jk}\frac{\beta_{0jk}}{\Vert\beta_{0jk}\Vert}I(\beta_{0jk}\ne0) +\Vert u_{jk}\Vert I(\beta_{0jk}=0)\right\}.
\end{equation*}

\end{proof}

\clearpage
\subsection{Additional numerical results}

\begin{table}[!htpb]
\caption{Simulation under the heterogeneity model:
unmatching rate=25$\%$ and nonzero regression coefficients under case 2. In each cell, the first row is number of true positives (standard deviation), and the second row is model size (standard deviation).}
\label{Tab:07}
\centering 
{%
\begin{tabular}{ccccc}
\hline
Correlation & \multicolumn{1}{c}{GLasso} & \multicolumn{3}{c}{Proposed} \\
\cline{2-5}
   &&\multicolumn{1}{c}{$\gamma=0.5$}&\multicolumn{1}{c}{$\gamma=0.7$}
   &\multicolumn{1}{c}{$\gamma=0.9$}\\
\hline
AR $\rho=$0.2&3.8(2.1)&2.4(4.2)&4.5(4.6)&5.4(3.4)\\
&32.7(17.5)&3.5(6.2)&7.9(7.7)&18.0(9.9)\\
AR $\rho=$0.5&6.0(2.4)&9.4(4.1)&10.2(3.2)&10.3(2.5)\\
&36.7(16.1)&13.3(5.9)&14.7(4.7)&20.0(7.7)\\
AR $\rho=$0.8&8.0(2.4)&11.7(0.9)&11.8(0.6)&11.7(0.8)\\
&46.9(18.5)&15.3(2.3)&15.6(3.0)&19.4(6.2)\\
Banded 1&3.8(2.1)&2.9(4.5)&4.4(4.5)&5.8(3.4)\\
&31.8(17.0)&4.3(6.6)&7.0(7.0)&18.1(9.9)\\
Banded 2&5.2(2.4)&9.2(4.5)&9.8(3.7)&9.5(3.0)\\
&32.7(18.8)&13.2(6.3)&14.4(5.3)&20.8(7.7)\\
Banded 3&5.7(2.5)&9.4(4.1)&9.7(3.4)&10.4(2.3)\\
&35.1(17.5)&13.0(5.6)&14.7(5.2)&22.8(8.7)\\
\hline
\end{tabular}%
}
\end{table}

\begin{table}[!tpb]
\caption{Simulation under the heterogeneity model:
unmatching rate=50$\%$ and nonzero regression coefficients under case 2. In each cell, the first row is number of true positives (standard deviation), and the second row is model size (standard deviation).}
\label{Tab:08}
\centering 
{%
\begin{tabular}{ccccc}
\hline
Correlation & \multicolumn{1}{c}{Lasso}  & \multicolumn{3}{c}{Proposed} \\
\cline{2-5}
   & &\multicolumn{1}{c}{$\gamma=0.5$}&\multicolumn{1}{c}{$\gamma=0.7$}
   &\multicolumn{1}{c}{$\gamma=0.9$}\\
\hline
AR $\rho=$0.2&3.7(2.0)&1.8(3.3)&3.8(3.8)&6.1(2.5)\\
&30.3(16.9)&2.8(5.4)&7.6(7.8)&21.4(8.2)\\
AR $\rho=$0.5&5.3(2.6)&8.0(3.8)&9.4(2.7)&9.5(1.7)\\
&33.5(20.1)&13.6(7.8)&17.4(6.3)&22.9(6.7)\\
AR $\rho=$0.8&8.2(2.3)&10.6(2.2)&10.9(1.0)&11.0(1.3)\\
&47.8(19.2)&18.1(5.5)&18.3(3.7)&21.8(6.6)\\
Banded 1&3.9(2.0)&1.9(3.0)&3.4(3.6)&5.8(3.1)\\
&31.2(18.5)&3.1(5.3)&6.3(7.0)&18.5(9.0)\\
Banded 2&5.1(2.4)&6.4(4.0)&8.6(3.0)&8.6(2.8)\\
&31.6(15.6)&10.4(7.4)&15.7(6.8)&21.7(8.8)\\
Banded 3&5.8(2.4)&8.3(3.5)&9.0(2.7)&9.4(2.1)\\
&36.2(18.5)&13.9(7.4)&15.6(5.4)&21.7(7.2)\\
\hline
\end{tabular}%
}
\end{table}

\begin{table}[!htpb]
\caption{Simulation under the homogeneity model: nonzero regression coefficients under case 1. In each cell, the first row is number of true positives (standard deviation), and the second row is model size (standard deviation).}
\label{Tab:01}
\centering 
{%
\begin{tabular}{ccccc}
\hline
Correlation & \multicolumn{1}{c}{GLasso}  & \multicolumn{3}{c}{Proposed} \\
\cline{2-5}
   &&\multicolumn{1}{c}{$\gamma=0.5$}&\multicolumn{1}{c}{$\gamma=0.7$}
   &\multicolumn{1}{c}{$\gamma=0.9$}\\
\hline
AR $\rho=$0.2&5.4(2.3)&9.5(4.0)&9.5(4.2)&9.1(3.8)\\
&37.6(17.2)&9.8(3.8)&10.5(4.3)&16.9(8.1)\\
AR $\rho=$0.5&7.1(2.7)&11.6(1.9)&11.6(1.6)&11.5(1.7)\\
&41.8(21.5)&11.8(1.4)&11.8(1.3)&15.2(4.4)\\
AR $\rho=$0.8&9.7(2.4)&11.9(0.4)&12.0(0.0)&12.0(0.0)\\
&48.3(17.8)&12.0(0.4)&12.1(0.6)&14.0(2.1)\\
Banded 1&5.1(2.0)&9.9(4.1)&10.1(3.8)&9.9(3.2)\\
&33.8(16.2)&10.2(3.9)&11.2(3.6)&17.9(6.9)\\
Banded 2&7.5(2.7)&12.0(0.0)&12.0(0.0)&11.6(1.1)\\
&45.5(19.3)&12.0(0.0)&12.3(1.1)&17.1(5.5)\\
Banded 3&8.3(2.3)&11.8(0.9)&11.9(0.8)&11.9(0.5)\\
&50.5(18.5)&11.8(0.9)&12.0(0.9)&15.9(3.7)\\
\hline
\end{tabular}%
}
\end{table}

\begin{table}[!tpb]
\caption{Simulation under the homogeneity model: nonzero regression coefficients under case 2. In each cell, the first row is number of true positives (standard deviation), and the second row is model size (standard deviation).}
\label{Tab:02}
\centering 
{%
\begin{tabular}{ccccc}
\hline
Correlation & \multicolumn{1}{c}{GLasso}  & \multicolumn{3}{c}{Proposed} \\
\cline{2-5}
   &&\multicolumn{1}{c}{$\gamma=0.5$}&
   \multicolumn{1}{c}{$\gamma=0.7$}&\multicolumn{1}{c}{$\gamma=0.9$}\\
\hline
AR $\rho=$0.2&4.1(2.4)&6.0(5.1)&6.9(4.8)&7.9(3.6)\\
&33.0(18.8)&6.2(5.1)&8.9(5.8)&19.9(7.7)\\
AR $\rho=$0.5&6.5(2.5)&11.2(2.3)&11.4(1.7)&11.3(1.6)\\
&39.2(18.2)&11.6(2.1)&12.4(2.0)&18.3(6.5)\\
AR $\rho=$0.8&8.2(2.6)&11.9(0.4)&11.9(0.4)&11.9(0.6)\\
&46.9(21.9)&12.0(0.6)&12.4(1.5)&16.4(5.8)\\
Banded 1&4.0(2.2)&5.3(5.3)&7.0(4.7)&7.8(3.5)\\
&33.5(16.2)&5.5(5.3)&8.1(5.3)&19.6(7.9)\\
Banded 2&6.0(2.2)&10.6(3.6)&11.1(2.5)&11.2(1.8)\\
&38.1(16.8)&10.9(3.2)&12.8(3.6)&19.1(6.2)\\
Banded 3&6.5(2.1)&11.4(1.9)&11.2(2.0)&11.3(2.1)\\
&39.9(16.8)&11.7(1.3)&12.5(2.1)&19.3(5.7)\\
\hline
\end{tabular}%
}
\end{table}

\begin{table}[!tpb]
\caption{SNP-level estimates using the proposed approach: DLBCL.}
\label{Tab:02}
\centering 
{%
\begin{tabular}{lclclclc}
\hline
SNP & Est & SNP & Est & SNP & Est & SNP & Est \\ \hline
ALOX15B\_01	&	0.00847	&	LMAN1\_02	&	-0.00386	&	MCP\_02	&	-0.00366	&	 NCF4\_42	&	-0.00429	\\
ALOX15B\_03	&	0.00238	&	LMAN1\_04	&	-0.00017	&	MCP\_04	&	0.01080	&	 NCF4\_43	&	-0.00424	\\
ALOX15B\_04	&	0.00397	&	LMAN1\_05	&	0.00876	&	MCP\_05	&	-0.00116	&	 NCF4\_44	&	0.00086	\\
ALOX15B\_06	&	0.00490	&	LMAN1\_06	&	-0.00343	&	MCP\_06	&	0.00242	&	 NCF4\_45	&	-0.00098	\\
ALOX15B\_07	&	0.00288	&	LMAN1\_07	&	0.00562	&	MCP\_07	&	0.00031	&	 NCF4\_46	&	-0.00055	\\
ALOX5\_01	&	0.00357	&	LMAN1\_08	&	0.00562	&	MCP\_08	&	-0.00511	&	 NCF4\_49	&	0.00130	\\
ALOX5\_41	&	0.00559	&	LMAN1\_09	&	0.00645	&	MEFV\_01	&	0.00544	&	 SERPINB3\_01	&	-0.00410	\\
ALOX5\_42	&	0.00033	&	MBP\_02	&	-0.00334	&	MEFV\_02	&	0.01260	&	 SERPINB3\_02	&	0.00262	\\
ALOX5\_43	&	-0.00611	&	MBP\_03	&	0.00084	&	MEFV\_03	&	-0.00468	&	 SERPINB3\_05	&	0.00489	\\
ALOX5\_44	&	0.00825	&	MBP\_04	&	0.00013	&	MEFV\_04	&	0.00699	&	 SERPINB3\_06	&	0.00177	\\
ALOX5\_45	&	0.00357	&	MBP\_05	&	0.00058	&	MEFV\_05	&	-0.00081	&	 STAT4\_08	&	-0.00873	\\
ALOX5\_46	&	-0.00093	&	MBP\_06	&	-0.00052	&	MIF\_01	&	0.00760	&	 STAT4\_09	&	0.00045	\\
ALOX5\_47	&	0.00428	&	MBP\_07	&	0.00058	&	MIF\_01\_2	&	0.00760	&	 STAT4\_10	&	0.00004	\\
ALOX5\_48	&	-0.00018	&	MBP\_08	&	0.00171	&	MIF\_14	&	-0.00358	&	 STAT4\_11	&	0.00485	\\
ALOX5\_49	&	0.00119	&	MBP\_09	&	-0.00075	&	MIF\_15	&	-0.00805	&	 STAT4\_12	&	-0.00138	\\
ALOX5\_51	&	0.00556	&	MBP\_10	&	-0.00100	&	MIF\_16	&	-0.00861	&	 STAT4\_13	&	-0.00156	\\
ALOX5\_52	&	0.00349	&	MBP\_12	&	-0.00149	&	MIF\_18	&	0.00750	&	 STAT4\_14	&	0.00097	\\
ALOX5\_53	&	0.00113	&	MBP\_13	&	0.00118	&	MIF\_19	&	0.00431	&	STAT4\_16	 &	-0.00500	\\
ALOX5\_54	&	0.01233	&	MBP\_14	&	0.00006	&	MIF\_20	&	0.00473	&	STAT4\_17	 &	-0.00605	\\
ALOX5\_55	&	0.00042	&	MBP\_15	&	0.00003	&	MIF\_21	&	0.00865	&	STAT4\_18	 &	0.00191	\\
CSF2\_02	&	0.01655	&	MBP\_16	&	0.00123	&	MIF\_22	&	-0.00365	&	 STAT4\_19	&	-0.00114	\\
DEFB1\_01	&	-0.00675	&	MBP\_17	&	0.00021	&	MIF\_23	&	-0.00232	&	 STAT4\_21	&	0.00118	\\
DEFB1\_02	&	-0.01199	&	MBP\_18	&	0.00032	&	MUC6\_01	&	-0.00386	&	 STAT4\_23	&	-0.00751	\\
DEFB1\_03	&	-0.00798	&	MBP\_19	&	-0.00094	&	MUC6\_02	&	-0.00825	 &	STAT4\_24	&	0.00048	\\
DEFB1\_04	&	0.00948	&	MBP\_20	&	-0.00028	&	MUC6\_03	&	-0.00041	&	 STAT4\_25	&	0.00687	\\
DEFB1\_05	&	-0.01192	&	MBP\_21	&	-0.00084	&	MUC6\_04	&	-0.00145	 &	STAT4\_29	&	0.00242	\\
DEFB1\_06	&	-0.00488	&	MBP\_22	&	-0.00023	&	MUC6\_07	&	-0.00329	 &	STAT4\_30	&	0.00274	\\
DEFB1\_07	&	-0.00561	&	MBP\_23	&	-0.00159	&	MUC6\_08	&	-0.00569	 &	STAT4\_31	&	-0.00531	\\
DEFB1\_08	&	-0.00195	&	MBP\_24	&	-0.00097	&	MUC6\_09	&	-0.00362	 &	STAT4\_33	&	-0.00654	\\
DEFB1\_09	&	0.00809	&	MBP\_25	&	0.00063	&	MUC6\_10	&	-0.01415	&	 STAT4\_34	&	0.00814	\\
DEFB1\_11	&	-0.00332	&	MBP\_26	&	0.00116	&	MUC6\_13	&	0.01833	&	 STAT4\_35	&	0.00020	\\
DEFB1\_12	&	-0.00140	&	MBP\_27	&	-0.00045	&	MUC6\_14	&	-0.01229	 &	STAT4\_36	&	0.00369	\\
DEFB1\_13	&	-0.01077	&	MBP\_28	&	0.00022	&	NCF4\_12	&	-0.00118	&	 STAT4\_37	&	0.00462	\\
IL10\_01	&	-0.00009	&	MBP\_29	&	-0.00214	&	NCF4\_18	&	-0.00292	 &	STAT4\_38	&	0.00229	\\
IL10\_02	&	-0.00009	&	MBP\_30	&	0.00030	&	NCF4\_33	&	-0.00254	&	 STAT4\_39	&	-0.00357	\\
IL10\_03	&	0.00000	&	MBP\_31	&	-0.00031	&	NCF4\_34	&	-0.00077	&	 STAT4\_41	&	0.00203	\\
IL10\_06	&	0.00001	&	MBP\_32	&	-0.00140	&	NCF4\_35	&	-0.00932	&	 STAT4\_42	&	0.00206	\\
IL10\_07	&	0.00001	&	MBP\_33	&	-0.00060	&	NCF4\_36	&	-0.00220	&	 STAT4\_43	&	0.00346	\\
IL10\_17	&	0.00000	&	MBP\_34	&	-0.00040	&	NCF4\_37	&	-0.00483	&	 STAT4\_44	&	0.00711	\\
IL10\_17\_2	&	0.00000	&	MBP\_35	&	0.00255	&	NCF4\_38	&	-0.00229	&	 STAT4\_45	&	0.00061	\\
LMAN1\_01	&	0.00918	&	MBP\_36	&	-0.00051	&	NCF4\_39	&	0.00575	&	 STAT4\_46	&	0.00687	\\
\hline
\end{tabular}%
}
\end{table}

\begin{table}[!tpb]
\caption{SNP-level estimates using the proposed approach: FL.}
\label{Tab:02}
\centering 
{%
\begin{tabular}{lclclclc}
\hline
SNP & Est & SNP & Est & SNP & Est & SNP & Est \\ \hline
ALOX12\_06	&	0.01078	&	IRAK2\_20	&	0.00404	&	MBP\_23	&	0.00021	&	 NCF4\_46	&	0.00449	\\
ALOX12\_07	&	0.01315	&	IRAK2\_21	&	0.00257	&	MBP\_24	&	-0.00002	&	 NCF4\_49	&	0.00383	\\
ALOX12\_09	&	0.00692	&	IRAK2\_22	&	0.00315	&	MBP\_25	&	0.00075	&	 PTK9L\_01	&	0.00275	\\
CLCA1\_01	&	-0.00360	&	IRAK2\_23	&	0.00319	&	MBP\_26	&	-0.00051	&	 PTK9L\_02	&	0.01068	\\
CLCA1\_02	&	0.00514	&	IRAK2\_24	&	-0.00416	&	MBP\_27	&	-0.00013	&	 PTK9L\_03	&	0.00498	\\
CLCA1\_03	&	-0.00255	&	IRAK2\_25	&	-0.00099	&	MBP\_28	&	0.00049	&	 SOD3\_27	&	0.00399	\\
CLCA1\_04	&	-0.00307	&	IRAK2\_26	&	-0.00563	&	MBP\_29	&	-0.00156	 &	STAT4\_08	&	-0.00143	\\
CLCA1\_05	&	0.00312	&	IRAK2\_27	&	0.00227	&	MBP\_30	&	0.00080	&	 STAT4\_09	&	0.00123	\\
CLCA1\_06	&	-0.00192	&	IRAK2\_28	&	0.00690	&	MBP\_31	&	-0.00148	&	 STAT4\_10	&	0.00068	\\
CLCA1\_07	&	0.00004	&	LIG4\_02	&	-0.01321	&	MBP\_32	&	0.00037	&	 STAT4\_11	&	-0.00025	\\
CLCA1\_08	&	0.00606	&	LMAN1\_01	&	-0.00644	&	MBP\_33	&	-0.00061	&	 STAT4\_12	&	0.00048	\\
CLCA1\_09	&	-0.00514	&	LMAN1\_02	&	0.00382	&	MBP\_34	&	-0.00097	&	 STAT4\_13	&	0.00096	\\
CLCA1\_10	&	0.00190	&	LMAN1\_04	&	-0.00467	&	MBP\_35	&	-0.00030	&	 STAT4\_14	&	0.00050	\\
CLCA1\_11	&	-0.00208	&	LMAN1\_05	&	0.00560	&	MBP\_36	&	0.00067	&	 STAT4\_16	&	0.00027	\\
CLCA1\_12	&	0.00447	&	LMAN1\_06	&	0.00355	&	MIF\_01	&	-0.00031	&	 STAT4\_17	&	-0.00088	\\
CLCA1\_13	&	0.00266	&	LMAN1\_07	&	-0.01116	&	MIF\_01\_2	&	-0.00031	 &	STAT4\_18	&	0.00121	\\
CLCA1\_15	&	-0.00221	&	LMAN1\_08	&	-0.01116	&	MIF\_14	&	0.00018	&	 STAT4\_19	&	0.00057	\\
CLCA1\_16	&	0.00426	&	LMAN1\_09	&	-0.00280	&	MIF\_15	&	0.00047	&	 STAT4\_21	&	0.00074	\\
CLCA1\_17	&	-0.00246	&	MBP\_02	&	0.00000	&	MIF\_16	&	0.00045	&	 STAT4\_23	&	-0.00026	\\
CLCA1\_18	&	-0.00427	&	MBP\_03	&	0.00016	&	MIF\_18	&	-0.00030	&	 STAT4\_24	&	0.00117	\\
CLCA1\_19	&	-0.00243	&	MBP\_04	&	0.00065	&	MIF\_19	&	0.00009	&	 STAT4\_25	&	-0.00119	\\
CLCA1\_20	&	0.00698	&	MBP\_05	&	0.00047	&	MIF\_20	&	0.00015	&	STAT4\_29	 &	0.00201	\\
CLCA1\_21	&	0.00241	&	MBP\_06	&	-0.00001	&	MIF\_21	&	-0.00031	&	 STAT4\_30	&	-0.00002	\\
CLCA1\_22	&	0.00865	&	MBP\_07	&	0.00016	&	MIF\_22	&	-0.00049	&	 STAT4\_31	&	-0.00118	\\
CLCA1\_23	&	-0.00595	&	MBP\_08	&	0.00008	&	MIF\_23	&	0.00021	&	 STAT4\_33	&	-0.00121	\\
IL17C\_01	&	0.01724	&	MBP\_09	&	0.00124	&	NCF4\_12	&	-0.00300	&	 STAT4\_34	&	-0.00174	\\
IRAK2\_01	&	-0.00111	&	MBP\_10	&	0.00053	&	NCF4\_18	&	0.00226	&	 STAT4\_35	&	0.00183	\\
IRAK2\_02	&	-0.00289	&	MBP\_12	&	-0.00062	&	NCF4\_33	&	-0.00048	 &	STAT4\_36	&	0.00046	\\
IRAK2\_10	&	0.00704	&	MBP\_13	&	0.00084	&	NCF4\_34	&	-0.00323	&	 STAT4\_37	&	-0.00041	\\
IRAK2\_11	&	0.00483	&	MBP\_14	&	-0.00016	&	NCF4\_35	&	-0.00107	&	 STAT4\_38	&	0.00059	\\
IRAK2\_12	&	-0.00195	&	MBP\_15	&	-0.00066	&	NCF4\_36	&	-0.00582	 &	STAT4\_39	&	0.00205	\\
IRAK2\_13	&	-0.00178	&	MBP\_16	&	0.00070	&	NCF4\_37	&	0.00349	&	 STAT4\_41	&	0.00105	\\
IRAK2\_14	&	-0.00234	&	MBP\_17	&	0.00017	&	NCF4\_38	&	-0.00065	&	 STAT4\_42	&	0.00011	\\
IRAK2\_15	&	0.00474	&	MBP\_18	&	-0.00010	&	NCF4\_39	&	0.00466	&	 STAT4\_43	&	0.00055	\\
IRAK2\_16	&	-0.00089	&	MBP\_19	&	-0.00121	&	NCF4\_42	&	0.00165	&	 STAT4\_44	&	0.00148	\\
IRAK2\_17	&	-0.00261	&	MBP\_20	&	0.00020	&	NCF4\_43	&	0.00013	&	 STAT4\_45	&	0.00062	\\
IRAK2\_18	&	-0.00047	&	MBP\_21	&	0.00101	&	NCF4\_44	&	-0.00265	&	 STAT4\_46	&	-0.00135	\\
IRAK2\_19	&	-0.00278	&	MBP\_22	&	-0.00004	&	NCF4\_45	&	-0.00420	 &		&		\\
\hline
\end{tabular}%
}
\end{table}

\begin{table}[!tpb]
\caption{SNP-level estimates using the proposed approach: CLL/SLL.}
\label{Tab:02}
\centering 
{%
\begin{tabular}{lclclclc}
\hline
SNP & Est & SNP & Est & SNP & Est & SNP & Est \\ \hline
ALOX5\_01	&	0.00540	&	IRAK2\_13	&	-0.00296	&	MBP\_14	&	-0.00157	&	 STAT4\_12	&	-0.00348	\\
ALOX5\_41	&	-0.00105	&	IRAK2\_14	&	-0.00353	&	MBP\_15	&	-0.01136	 &	STAT4\_13	&	-0.00234	\\
ALOX5\_42	&	0.00280	&	IRAK2\_15	&	0.00202	&	MBP\_16	&	-0.00984	&	 STAT4\_14	&	-0.00177	\\
ALOX5\_43	&	0.00393	&	IRAK2\_16	&	-0.00217	&	MBP\_17	&	0.00380	&	 STAT4\_16	&	0.00067	\\
ALOX5\_44	&	0.00407	&	IRAK2\_17	&	0.00271	&	MBP\_18	&	-0.01136	&	 STAT4\_17	&	-0.00307	\\
ALOX5\_45	&	0.00504	&	IRAK2\_18	&	0.00186	&	MBP\_19	&	-0.00883	&	 STAT4\_18	&	-0.00050	\\
ALOX5\_46	&	-0.00096	&	IRAK2\_19	&	0.00266	&	MBP\_20	&	-0.00373	&	 STAT4\_19	&	-0.00249	\\
ALOX5\_47	&	0.00234	&	IRAK2\_20	&	-0.00165	&	MBP\_21	&	0.00081	&	 STAT4\_21	&	-0.00214	\\
ALOX5\_48	&	0.00200	&	IRAK2\_21	&	-0.00307	&	MBP\_22	&	-0.00438	&	 STAT4\_23	&	-0.00306	\\
ALOX5\_49	&	0.00334	&	IRAK2\_22	&	-0.00025	&	MBP\_23	&	0.00720	&	 STAT4\_24	&	-0.00144	\\
ALOX5\_51	&	-0.00205	&	IRAK2\_23	&	0.00186	&	MBP\_24	&	0.00087	&	 STAT4\_25	&	0.00143	\\
ALOX5\_52	&	-0.00043	&	IRAK2\_24	&	0.00226	&	MBP\_25	&	-0.00234	&	 STAT4\_29	&	-0.00179	\\
ALOX5\_53	&	0.00556	&	IRAK2\_25	&	-0.00147	&	MBP\_26	&	0.00109	&	 STAT4\_30	&	-0.00187	\\
ALOX5\_54	&	0.00552	&	IRAK2\_26	&	-0.00364	&	MBP\_27	&	0.00462	&	 STAT4\_31	&	-0.00045	\\
ALOX5\_55	&	-0.00058	&	IRAK2\_27	&	0.00008	&	MBP\_28	&	0.00801	&	 STAT4\_33	&	-0.00056	\\
IL10\_01	&	-0.00305	&	IRAK2\_28	&	-0.00263	&	MBP\_29	&	-0.00328	 &	STAT4\_34	&	-0.00119	\\
IL10\_02	&	-0.00421	&	MBP\_02	&	0.00696	&	MBP\_30	&	0.00278	&	 STAT4\_35	&	0.00023	\\
IL10\_03	&	-0.00313	&	MBP\_03	&	0.01177	&	MBP\_31	&	-0.00510	&	 STAT4\_36	&	-0.00246	\\
IL10\_06	&	-0.00212	&	MBP\_04	&	0.01112	&	MBP\_32	&	0.00658	&	 STAT4\_37	&	0.00073	\\
IL10\_07	&	-0.00313	&	MBP\_05	&	0.00279	&	MBP\_33	&	-0.00924	&	 STAT4\_38	&	-0.00174	\\
IL10\_17	&	-0.01264	&	MBP\_06	&	0.01161	&	MBP\_34	&	0.00411	&	 STAT4\_39	&	-0.00032	\\
IL10\_17\_2	&	-0.01264	&	MBP\_07	&	-0.00203	&	MBP\_35	&	0.00028	&	 STAT4\_41	&	-0.00224	\\
IRAK2\_01	&	0.00421	&	MBP\_08	&	-0.00272	&	MBP\_36	&	0.00618	&	 STAT4\_42	&	-0.00210	\\
IRAK2\_02	&	0.00072	&	MBP\_09	&	0.00394	&	STAT4\_08	&	-0.00153	&	 STAT4\_43	&	-0.00111	\\
IRAK2\_10	&	-0.00098	&	MBP\_10	&	-0.00098	&	STAT4\_09	&	-0.00124	 &	STAT4\_44	&	-0.00281	\\
IRAK2\_11	&	-0.00096	&	MBP\_12	&	0.00964	&	STAT4\_10	&	-0.00123	&	 STAT4\_45	&	0.00067	\\
IRAK2\_12	&	0.00107	&	MBP\_13	&	0.00312	&	STAT4\_11	&	-0.00205	&	 STAT4\_46	&	0.00143	\\
\hline
\end{tabular}%
}
\end{table}

\begin{center}
\scriptsize
\begin{longtable}{ccccccc}
\caption[Analysis of the NHL data using group Lasso]{Analysis of the NHL data using group Lasso for each subtype separately: $L_2$-norm of estimate for a specific gene; OOI: observed occurrence index.} \label{s_glasso} \\
\hline \multicolumn{1}{c}{Gene} & \multicolumn{2}{c}{DLBCL} & \multicolumn{2}{c}{FL} & \multicolumn{2}{c}{CLL/SLL} \\
\cline{2-7}
   & $L_2$-norm & OOI & $L_2$-norm & OOI & $L_2$-norm &OOI\\
\hline
\endfirsthead

\multicolumn{7}{c}%
{{\bfseries \tablename\ \thetable{} -- continued from previous page}} \\
\hline \multicolumn{1}{c}{Gene} &
\multicolumn{2}{c}{DLBCL} &
\multicolumn{2}{c}{FL} &
\multicolumn{2}{c}{CLL/SLL}\\
\cline{2-7}
   & $L_2$ norm& OOI & $L_2$ norm & OOI & $L_2$ norm &OOI\\
\hline

\endhead

\hline \multicolumn{7}{r}{{Continued on next page}} \\ \hline
\endfoot

\hline
\endlastfoot
AHR	&		&		&	0.017	&	0.97	&		&		\\
ALOX12	&		&		&	0.019	&	1.00	&		&		\\
ALOX15B	&	0.012	&	1.00	&		&		&		&		\\
APEX1	&	0.000	&	0.33	&		&		&		&		\\
BAT5	&		&		&	0.003	&	0.33	&		&		\\
BHMT	&	0.003	&	0.74	&	0.001	&	0.28	&		&		\\
C1QA	&	0.010	&	0.91	&		&		&		&		\\
C1QB	&		&		&	0.019	&	0.90	&		&		\\
C1QG	&	0.003	&	0.46	&		&		&		&		\\
C1S	&		&		&	0.004	&	0.56	&		&		\\
C2	&		&		&		&		&	0.014	&	0.90	\\
C4BPA	&	0.001	&	0.25	&		&		&		&		\\
C5	&	0.010	&	0.93	&		&		&		&		\\
C8A	&		&		&	3E-04	&	0.36	&		&		\\
CASP10	&	0.010	&	0.97	&		&		&		&		\\
CCL13	&	0.004	&	0.84	&		&		&		&		\\
CCND1	&		&		&	0.003	&	0.46	&		&		\\
CCR1	&		&		&	0.008	&	0.79	&		&		\\
CCR8	&		&		&		&		&	0.018	&	0.76	\\
CENTA1	&		&		&		&		&	0.008	&	0.66	\\
CSF2	&	0.006	&	0.88	&		&		&	0.006	&	0.69	\\
CSF3	&	0.011	&	0.94	&		&		&		&		\\
CTNNB1	&		&		&	0.004	&	0.54	&		&		\\
CX3CR1	&	0.006	&	0.81	&		&		&		&		\\
CYP1A2	&		&		&	0.002	&	0.51	&		&		\\
CYP1B1	&	0.012	&	0.91	&	0.001	&	0.23	&		&		\\
CYP2C9	&		&		&	0.017	&	0.90	&		&		\\
DEF6	&		&		&	0.006	&	0.69	&		&		\\
DEFB1	&	0.002	&	0.70	&		&		&		&		\\
DHX33	&	0.015	&	1.00	&		&		&		&		\\
EPHX1	&	0.007	&	0.88	&	0.006	&	0.67	&		&		\\
ERCC2	&	0.010	&	0.97	&		&		&		&		\\
ERCC5	&		&		&	0.003	&	0.18	&		&		\\
GGH	&		&		&	0.001	&	0.31	&	0.003	&	0.45	\\
ICAM2	&	0.009	&	0.90	&		&		&		&		\\
ICAM4	&	0.004	&	0.86	&		&		&		&		\\
ICAM5	&		&		&	0.011	&	0.95	&		&		\\
IFNGR1	&		&		&		&		&	0.002	&	0.38	\\
IL10	&	0.002	&	0.61	&		&		&	0.011	&	0.86	\\
IL10RA	&	0.005	&	0.84	&		&		&		&		\\
IL15	&	0.007	&	0.96	&		&		&	0.005	&	0.90	\\
IL15RA	&		&		&	0.002	&	0.54	&		&		\\
IL17C	&		&		&	0.013	&	0.92	&		&		\\
IL3	&	0.004	&	0.84	&		&		&		&		\\
KLK6	&		&		&		&		&	0.006	&	0.38	\\
LEPR	&	0.007	&	0.90	&		&		&		&		\\
LIG4	&		&		&		&		&	0.000	&	0.17	\\
LMAN1	&		&		&	0.002	&	0.41	&		&		\\
LPO	&	0.005	&	0.61	&		&		&		&		\\
MASP2	&		&		&	0.004	&	0.62	&		&		\\
MCP	&	0.011	&	0.90	&		&		&		&		\\
MEFV	&	4E-04	&	0.38	&		&		&		&		\\
MIF	&	0.002	&	0.43	&		&		&		&		\\
MLH1	&	0.001	&	0.54	&		&		&		&		\\
MTHFD2	&	0.004	&	0.52	&		&		&		&		\\
MTR	&		&		&		&		&	0.001	&	0.28	\\
MTRR	&	0.014	&	1.00	&		&		&		&		\\
MUC6	&	0.008	&	0.96	&		&		&		&		\\
MUC7	&		&		&	0.005	&	0.77	&		&		\\
MYC	&		&		&		&		&	0.006	&	0.86	\\
T2	&		&		&		&		&	0.001	&	0.34	\\
NBS1	&	0.001	&	0.65	&		&		&		&		\\
NCF4	&		&		&	0.001	&	0.18	&		&		\\
NFKBIE	&	0.007	&	0.94	&		&		&	0.003	&	0.41	\\
NOS3	&	0.006	&	0.84	&		&		&		&		\\
OGG1	&	0.003	&	0.64	&	0.010	&	0.72	&		&		\\
PDCD10	&	0.001	&	0.52	&		&		&		&		\\
PFKFB2	&	0.005	&	0.87	&		&		&		&		\\
PPARG	&		&		&	0.005	&	0.38	&		&		\\
PRO1580	&		&		&		&		&	0.005	&	0.55	\\
PTK9	&	0.011	&	1.00	&		&		&		&		\\
PTK9L	&		&		&	0.018	&	0.92	&		&		\\
RAD23B	&	0.019	&	1.00	&		&		&		&		\\
RAG1	&	0.004	&	0.74	&		&		&		&		\\
SECTM1	&	0.001	&	0.54	&		&		&		&		\\
SELE	&	0.002	&	0.75	&		&		&	0.001	&	0.45	\\
SENP3	&	0.004	&	0.75	&		&		&		&		\\
SERPINB3	&	0.004	&	0.84	&	0.003	&	0.26	&		&		\\
SOCS4	&	0.009	&	0.93	&		&		&		&		\\
SOD3	&	3E-04	&	0.28	&	0.008	&	0.85	&		&		\\
STAT5B	&		&		&		&		&	0.002	&	0.52	\\
STAT6	&	0.002	&	0.57	&		&		&		&		\\
STK11	&		&		&		&		&	0.002	&	0.48	\\
TCN1	&	0.008	&	0.94	&		&		&		&		\\
TICAM1	&	0.001	&	0.59	&		&		&		&		\\
TLR1	&		&		&		&		&	0.019	&	0.90	\\
TLR9	&	0.002	&	0.42	&		&		&		&		\\
TNFRSF18	&		&		&	0.003	&	0.67	&		&		\\
TOLLIP	&	8E-05	&	0.29	&	0.005	&	0.54	&		&		\\
TP53	&		&		&	0.001	&	0.46	&		&		\\
TYK2	&	0.003	&	0.81	&		&		&		&		\\
WDHD1	&	0.015	&	1.00	&		&		&		&		\\
WRN	&		&		&	0.003	&	0.49	&		&		\\
XRCC1	&	0.004	&	0.70	&		&		&	9E-05	&	0.28	\\
XRCC3	&	4E-04	&	0.39	&		&		&		&		\\
XRCC4	&		&		&		&		&	0.010	&	0.76	\\
ZNF76	&		&		&		&		&	0.006	&	0.66	\\
ZP1	&		&		&		&		&	0.006	&	0.76	\\
\hline
\end{longtable}
\end{center}

\end{document}